\documentclass[10pt]{amsart}

\usepackage{mathrsfs, amsthm, amsmath}

\usepackage{bbm}

\usepackage{graphicx}
\usepackage{pstricks,pst-eucl}

\usepackage{tikz}
\usetikzlibrary{decorations.fractals}

\newtheorem{theorem}{Theorem}[section]
\newtheorem{lemma}[theorem]{Lemma}
\newtheorem{cor}[theorem]{Corollary}
\newtheorem{prop}[theorem]{Proposition}

\theoremstyle{definition}
\newtheorem{defn}[theorem]{Definition}
\newtheorem{example}[theorem]{Example}

\theoremstyle{remark}
\newtheorem{remark}[theorem]{Remark}

\numberwithin{equation}{section}

\DeclareMathAlphabet      {\mathbfit}{OML}{cmm}{b}{it}

\let\text=\mbox

\catcode`\@=11

\renewcommand{\a}{\alpha}
\renewcommand{\b}{\beta}
\renewcommand{\c}{\gamma}
\renewcommand{\d}{\delta}
\newcommand{\g}{\lambda}
\renewcommand{\o}{\omega}
\newcommand{\q}{\quad}

\newcommand{\s}{\sigma}

\newcommand{\cal}{\mathcal}

\newcommand{\M}{{\cal M}}

\newcommand{\ty}{\infty}

\newcommand{\f}{\varphi}
\newcommand{\ov}[1]{\overline{#1}}

\renewcommand{\O}{\Omega}
\newcommand{\pa}{\partial}

\newcommand{\st}{\subset}
\newcommand{\stq}{\subseteq}

\newcount\br@j
\newcommand{\udesno}[1]{\unskip\nobreak\hfil\penalty50\hskip1em\hbox{}
             \nobreak\hfil{#1\unskip\ignorespaces}
                 \parfillskip=\z@ \finalhyphendemerits=\z@\par
                 \parfillskip=0pt plus 1fil}
\catcode`\@=11

\newcommand{\eR}{\mathbb{R}}
\newcommand{\eN}{\mathbb{N}}
\newcommand{\Ze}{\mathbb{Z}}

\newcommand{\Ce}{\mathbb{C}}

\newcommand{\re}{\mathop{\mathrm{Re}}}

\newcommand{\po}{{\mathop{\mathcal P}}}

\newcommand{\res}{\operatorname{res}}

\newcommand{\sideremark}[1]{\ifvmode\leavevmode\fi\vadjust{\vbox to0pt{\vss 
      \hbox to 0pt{\hskip\hsize\hskip1em           
 \vbox{\hsize2cm\tiny\raggedright\pretolerance10000
 \noindent #1\hfill}\hss}\vbox to8pt{\vfil}\vss}}}%

\newcommand{\D}{\mathrm{d}}
\newcommand{\I}{\mathbbm{i}}
\newcommand{\E}{\mathrm{e}}

\begin{document}

\title[Complex dimensions of fractals and meromorphic extensions]{Complex dimensions of fractals and meromorphic extensions of fractal zeta functions}%

\author{Michel L.\ Lapidus}
\address{Department of Mathematics,
University of California, Riverside, California 92521-0135 USA}
\email{lapidus@math.ucr.edu}

\thanks{\vskip1mm The research of Michel L.~Lapidus was partially supported by the National Science
Foundation under grants ~DMS-0707524 and DMS-1107750, as well as by the Institut des Hautes \' Etudes Scientifiques (IH\' ES) where the first author was a visiting professor in the Spring of 2012 while part of this research was completed. The research of Goran Radunovi\'c and Darko \v Zubrini\'c was supported in part by the Croatian Science Foundation under the project IP-2014-09-2285 and by the Franco-Croatian 
PHC-COGITO project.}

\author{Goran Radunovi\'c}
\address{Department of Applied Mathematics, Faculty of Electrical Engineering and Computing, University of Zagreb, Unska 3, 10000 Zagreb, Croatia}
\email{goran.radunovic@fer.hr}

\author{Darko \v Zubrini\'c}
\address{Department of Applied Mathematics, Faculty of Electrical Engineering and Computing, University of Zagreb, Unska 3, 10000 Zagreb, Croatia}
\email{darko.zubrinic@fer.hr}



\renewcommand{\subjclassname}{%
\textup{2010} Mathematics Subject Classification}


\begin{abstract}
 We study 
meromorphic extensions of distance and tube zeta functions, as well as of geometric zeta functions of fractal strings.
The distance zeta function $\zeta_A(s):=\int_{A_\delta} d(x,A)^{s-N}\mathrm{d}x$, where $\d>0$ is fixed and $d(x,A)$ denotes the Euclidean distance from $x$ to $A$, has been introduced by the first author in 2009, extending the definition of the zeta function $\zeta_{\mathcal L}$ associated with
bounded fractal
strings $\mathcal L=(\ell_j)_{j\ge1}$ to arbitrary bounded subsets $A$ of the $N$-dimensional Euclidean space.
The abscissa of Lebesgue (i.e., absolute) convergence $D(\zeta_A)$ coincides with $D:=\overline\dim_BA$,  the upper box (or Minkowski) dimension of~$A$.
The (visible) complex dimensions of $A$ are the poles of the meromorphic continuation of the fractal zeta
function (i.e., the distance or tube zeta function) of $A$ to a suitable connected neighborhood of the ``critical line'' $\{\re s=D\}$.
We establish several meromorphic extension results, assuming some suitable information about the
second term of the asymptotic expansion of the tube function $|A_t|$ as $t\to0^+$, where $A_t$ is the Euclidean
$t$-neighborhood of $A$. 
We pay particular attention
to a class of Minkowski measurable sets, such that $|A_t|=t^{N-D}(\mathcal M+O(t^\c))$ as $t\to0^+$,
with $\c>0$, and to a class of Minkowski nonmeasurable sets, such that $|A_t|=t^{N-D}(G(\log t^{-1})+O(t^\c))$ as $t\to0^+$, where $G$ is a nonconstant periodic function and $\c>0$.
In both cases, we show that $\zeta_A$ can be meromorphically extended (at least) to the 
open right half-plane $\{\re s>D-\c\}$ and determine the corresponding visible complex dimensions. Furthermore, up to a multiplicative constant, the residue of $\zeta_A$ evaluated at $s=D$ is
shown to be equal to $\mathcal M$ (the Minkowski content of $A$) and to the mean value of $G$
(the average Minkowski content of $A$), respectively. 
Moreover, we construct a class of fractal strings with principal complex dimensions of any prescribed order, as well as with an infinite number of essential singularities on the critical line $\{\re s=D\}$.
Finally, using an appropriate quasiperiodic version of the above construction, with infinitely many
suitably chosen quasiperiods associated with a two-parameter family of generalized Cantor sets, 
we construct ``maximally-hyperfractal'' compact subsets of $\eR^N$, for $N\ge1$ arbitrary.
These are compact subsets of $\eR^N$ such that the corresponding fractal zeta functions have nonremovable
singularities at every point of the critical line $\{\re s=D\}$.
\end{abstract}

\maketitle

\tableofcontents

\section{Introduction}\label{intro}

\subsection{Motivations and goals}\label{motivation} This research is a continuation of our work,
initiated by the first author and started in \cite{dtzf} (see also \cite{fzf}), on extending the theory of zeta functions for fractal strings,
to fractal sets and arbitrary compact sets
in Euclidean spaces.
 The new zeta function on which it is based has been introduced
in 2009 by the first author; see its definition below in Eq.\  (\ref{z}). We denote this zeta function by $\zeta_A$ and refer to it as a ``distance zeta function". Here, by a fractal set,
we mean any (nonempty) bounded set $A$ of the Euclidean space $\eR^N$, with $N\geq 1$. Fractality refers to the fact that the notion of fractal dimension, and in particular, of the upper
Minkowski dimension of a bounded set (also called in the literature the upper box dimension, Bouligand dimension, or limit capacity, etc.) is a basic tool in the study of the properties of the associated zeta functions considered in this article, much as is the case in \cite{lapidusfrank12} for fractal strings (i.e., when $N=1$). 

More generally, as is shown throughout the higher-dimensional theory of complex dimensions developed in
[LapRa\v Zu1--8] and in the present paper, the notion of complex dimensions plays a key role in understanding
the geometric oscillations which are intrinsic to fractals. In fact, ``fractality'' is defined as the existence of a nonreal complex dimension (or else, the existence of a natural boundary for $\zeta_A$, beyond which $\zeta_A$ cannot be meromorphically extended). This is the same definition as in \cite[\S12.1 and \S13.4.3]{lapidusfrank12}, except for the fact that we now have in our possession a general definition and a well developed theory of fractal zeta functions, valid in any dimension; namely, the distance and tube zeta functions (see Definition \ref{defn} or Definition \ref{zeta_tilde}, respectively), which extend to arbitrary $N\ge 1$ and any bounded subset of $\eR^N$ the usual notion of geometric zeta function of a fractal string (see \S\ref{zeta_s}).
Accordingly, within the present higher-dimensional theory, the visible complex dimensions of a bounded subset $A$ of $\eR^N$ are now defined as the poles of the meromorphic continuation (when it exists) of $\zeta_A$ to a suitable connected open subset of $\Ce$; see Definition \ref{1.331/2}.

We next briefly discuss some aspects of the earlier work on fractal strings (or fractal sprays) motivating parts of the present work.

Along with several collaborators, the first author has undertaken since the early 1990s a systematic study of zeta functions associated with fractal strings and their counterparts in certain higher-dimensional situations, namely, fractal sprays;
see, in particular, [Lap1--2, HeLap, LapPo1--2, LapMa], 
and the books [Lap-vFr1--2]. 
The resulting developments have grown into a well-established theory 
of fractal strings, fractal sprays and their complex (fractal) dimensions,
and is today an active and rapidly growing area of research. 
For the theory of fractal strings and/or complex dimensions in a variety of situations, beside the aforementioned books and papers, see, for example, 
[DubSep, 
Es2,
EsLi1--2,
Fal2,
Fr,
KeKom, Kom,
LapPeWi,
LapRa\v Zu1--8,
L\'eMen,
MorSep,
MorSepVi,
Ol1--2,
Pe, 
PeWi,
Ra1--2,
RatWi,
Tep1--2], 
and the relevant references therein. In addition, we point out that Chapter 13 of \cite{lapidusfrank12} contains an exposition of several recent developments in the theory, prior to the present higher-dimensional theory of complex dimensions.

Throughout the monograph \cite{lapidusfrank12}, dealing with numerous aspects of the study of the geometric zeta functions of bounded fractal strings and their generalizations, 
it is assumed that the meromorphic extensions of the geometric zeta functions exist in a suitable region of the complex plane.
Furthermore,
the proof of \cite[Thm.\ 6.23]{lapidusfrank12}, dealing with the geometric zeta function of the $a$-string, shows that the construction of the meromorphic extensions can be quite delicate in certain situations. Moreover, the question of how far one can meromorphically extend the spectral zeta function associated with the (Dirichlet) Laplacian on a bounded open subset of $\eR^N$
is closely related to the value of the so-called inner Minkowski dimension of the boundary of the set; see [Lap1--2]
and \cite{brezish}. 
Therefore, we believe that the difficult problem of the existence and construction of meromorphic extensions of fractal zeta functions deserves a careful study.
In addition, it is an essential prerequisite for developing the (geometric) theory of complex dimensions in
$\eR^N$. Indeed, given a bounded set $A\st\eR^N$ and a connected open neighborhood $U$ of the critical line $\{\re s=D\}$, to which the distance (or the tube) zeta function of $A$, $\zeta_A$ (or $\tilde\zeta_A$),
can be meromorphically extended,\footnote{Here and in the rest of \S1.1, $D:=\ov\dim_BA$ denotes the upper Minkowski (or box) dimension of $A$; see Eq.\  \eqref{dim} in \S1.2 below. Also, the vertical line $\{\re s=D\}$ is referred to as the {\em critical line}.}
the {\em visible complex dimensions} of $A$ (relative to $U$)
are the poles located in $U$ of the (necessarily unique)
meromorphic extension of $\zeta_A$ (or $\tilde\zeta_A$).\footnote{Under mild assumptions, when $D<N$, the resulting visible complex dimensions are the same for either $\zeta_A$ or $\tilde\zeta_A$; see Remark \ref{1.3.121/2} in \S\ref{residues_m_distance} below.} 
Under the above hypotheses, the (upper) Minkowski dimension of $A$, $D:=\ov\dim_BA$, is a visible complex dimension of $A$ with maximal real part (since, according to a basic result in [LapRa\v Zu2] recalled below,
$\zeta_A$ and $\tilde\zeta_A$ are holomorphic in the open right half-plane $\{\re s>D\}$).

Another important (and related) motivation for the work in this paper is that the existence of a suitable meromorphic continuation (satisfying appropriate growth conditions) is a prerequisite for the development of fractal tube formulas (obtained in [LapRa\v Zu1,5]) 
in our general context (where we do not make any assumptions of self-similarity or ``self-alikeness'' and work with arbitrary bounded subsets of $\eR^N$, where $N\geq 1$)  and extending the corresponding fractal tube formulas obtained for fractal strings in \cite[Ch.\ 8]{lapidusfrank12} and for fractal sprays in \cite{lappewi1}.
In fact, it is noteworthy that some of the results obtained in \cite{cras2} show that our main result in this paper (Theorem \ref{measurable}) is, in some sense, optimal.

We certainly do not address or solve the general problem of meromorphic continuation here, but we consider two main classes of examples, in the Minkowski measurable case (Theorem \ref{measurable}) and in the non-Minkowski measurable case (Theorem \ref{nonmeasurable}), which are very often encountered in practice.
Using number-theoretic techniques, we also construct compact sets in $\eR^N$ whose associated fractal zeta functions have singularities at every point of the critical line.

We note that the results presented in this paper can be extended to the broader and more flexible framework of relative fractal drums, as described in our next paper \cite{rfds}, as well as for certain classes of fractals obeying a nonstandard power law (see [LapRa\v Zu5,8]). 
The interested reader will find in [LapRa\v Zu1,3,5--6] 
additional examples of the determination of the meromorphic extensions and the associated visible complex dimensions of bounded sets (and, more generally, of relative fractal drums) in $\eR^N$.
\medskip

We close this part of the introduction by mentioning that, as is well known, establishing the existence of the meromorphic continuations of arithmetic zeta functions and related Dirichlet series is one of the main challenges of analytic number theory; see, e.g., [Edw, Ser, ParsSh1--2, Es1, Lap-vFr2] and the relevant references therein. In the case of dynamical (or Ruelle) zeta functions, the counterpart of this problem has been tackled by W.\ Parry and M.\ Pollicott [ParrPol1--2], D.\ Ruelle [Rue1--3], and many other researchers. It would be interesting to suitably adapt the techniques developed in those references to our present geometric setting in order to enable us to deal with a broader class of situations. Indeed, the results obtained in the present paper certainly do not provide the last word concerning the meromorphic continuations of fractal zeta functions of bounded sets (or, more generally, of relative fractal drums). This problem is left for a future work by either the authors or the interested readers.
\medskip


Let us now briefly describe the contents of the paper:\newline
\q In \S\ref{notation}, we provide the necessary background material concerning the Minkowski (or box) dimension and the Minkowski content, along with some basic notation.

In \S\ref{ch_distance}, we recall from \cite{dtzf} 
the definition (Definition \ref{defn}) and some of the key properties (Theorem \ref{an}) of the distance zeta function, $\zeta_A$, of a bounded set $A\st\eR^N$, along with the definition of the abscissae of convergence, holomorphic continuation and meromorphic continuation of $\zeta_A$, denoted by
$D(\zeta_A)$, $D_{\rm hol}(\zeta_A)$ and $D_{\rm mer}(\zeta_A)$, respectively, as well as the notion of visible complex dimensions. We also briefly recall the definition of a fractal string $\mathcal L$ and the connection between the geometric zeta function of $\mathcal L$, denoted by $\zeta_{\mathcal L}$, and the distance zeta function of the boundary of any geometric realization of $\mathcal L$ by a bounded open subset of $\eR$.

In \S\ref{residues_m}, we first recall (in \S\ref{residues_m_distance}) the definition of the tube zeta function, $\tilde\zeta_A$, 
a counterpart of the distance zeta function $\zeta_A$ which is expressed in terms of the volume of the tubular neighborhoods of $A$  (see Definition \ref{zeta_tilde}),
and then provide the functional equation connecting these two fractal zeta functions, $\tilde\zeta_A$ and $\zeta_A$ (see Theorem \ref{equr} and Eq.\  \eqref{equ_tilde}).

The rest of \S\ref{residues_m_distance} is a brief overview of the results about 
the residues of fractal zeta functions, primarily those computed at the value of the fractal dimension, more precisely, at the value of
the upper Minkowski dimension $D$ of the corresponding fractal set $A$.
Namely, the residue at $D$ is closely related to the Minkowski content of the fractal,
provided the given fractal set $A$ is Minkowski measurable (see \S1.2) and the zeta function $\zeta_A$ or $\tilde\zeta_A$ has a meromorphic extension to a connected open neighborhood of $D$ in the complex plane.
If $A$ is not necessarily Minkowski measurable but is Minowski nondegenerate (i.e., $0<\M_*^D<\M^{*D}<\ty$, where $\M_*^D$ and $\M^{*D}$ denote, respectively, the lower and upper $D$-dimensional Minkowski content of $A$, see \S1.2), then the residue of $\zeta_A$ (or $\tilde\zeta_A$) at $D$ is shown to be squeezed between a positive multiple of $\M_*^D$ and $\M^{*D}$. (See Theorem \ref{pole1} and the discussion following it.)

In \S\ref{scarpet}, we then discuss in detail the important example of the Sierpi\'nski carpet $A\st\eR^2$, for which symmetry and scaling considerations help us determine the complex dimensions and the meromorphic continuation of $\zeta_A$ to all of~$\Ce$.


In \S\ref{merom_ext}, we provide an explicit construction of meromorphic extensions for a class of  Minkowski measurable subsets of $\eR^N$ (Theorem \ref{measurable}),
as well as for a class of Minkowski nonmeasurable sets of the lattice type (Theorem \ref{nonmeasurable});
we also determine the corresponding visible complex dimensions. 
Moreover, in \S\ref{gen}, we construct a class of self-similar fractal strings (in a more general sense) with principal complex dimensions of arbitrary multiplicities, and even with an infinite number of essential singularities along the critical line (see Theorem \ref{higher_order_dim} and the examples preceding it), by using iterated tensor products of fractal strings.

In \S\ref{hyperfractals}, we provide a construction of maximally hyperfractal bounded strings, as well as of bounded subsets of $\eR$, by means of a suitable infinite sequence of generalized Cantor sets depending on two auxilliary parameters; see Theorem \ref{hyper} of \S\ref{hyperfratalsr}. As a consequence, in Theorem \ref{mh} of \S\ref{hyperfratalsn},
we construct maximally hyperfractal subsets of $\eR^N$ of arbitrary Minkowski dimension in $(N-1,N)$, for any $N\ge1$. 
Recall that a bounded set $A\st\eR^N$ is said to be {\em maximally hyperfractal} if its associated fractal zeta function admits a nonremovable singularity at every point of the critical line $\{\re s=D\}$, where $D$ is the upper Minkowski dimension of $A$. 

\subsection{Basic notation and definitions}\label{notation} In the sequel, we use the following notation.
By $|E|_N$, we denote the $N$-dimensional Lebesgue measure of a measurable 
subset $E$ of $\eR^N$, where $N\ge1$.
(When no ambiguity may arise, we simply write $|E|$ instead of $|E|_N$.) 
Let $A$ be a bounded subset of $\eR^N$.
Given $r\ge0$, the {\em upper $r$-dimensional Minkowski content} $\M^{*r}(A)$ of $A$ is defined by 
\begin{equation}\label{mink}
\M^{*r}(A)=\limsup_{t\to 0^+}\frac{|A_t|}{t^{N-r}}, 
\end{equation}
and the {\em lower $r$-dimensional Minkowski content} of $A$, denoted by $\M_*^r(A)$, is defined analogously, with a lower limit instead of an upper limit in the counterpart of (\ref{mink}). 
Here and in the sequel, as in \S\ref{motivation},  $A_t:=\{x\in\eR^N:d(x,A)<t\}$ denotes the $t$-{\em neighborhood} (or {\em tubular neighborhood of radius} $t$) of $A$, and $d(x,A)$ is the Euclidean distance from $x$ to $A$.

The {\em upper Minkowski dimension} of $A$ is defined by
\begin{equation}\label{dim}
\ov\dim_BA:=\inf\{r>0:\M^{*r}(A)=0\}=\sup\{r>0:\M^{*r}(A)=+\ty\}.
\end{equation}
The lower Minkowski dimension of $A$, denoted by $\underline \dim_BA$,  is defined analogously, with $\M_*^r(A)$ instead of  $\M^{*r}(A)$ in the counterpart of (\ref{dim}). If both dimensions 
$\ov\dim_BA$ and $\underline\dim_BA$
coincide, their common value is denoted by $\dim_BA$, and is called the {\em Minkowski dimension} of $A$ (or Minkowski--Bouligand dimension), or else, simply, {\em box dimension}\label{MinkDim}. For general properties of box dimensions, see, e.g., \cite{falc}.

If there exists $D\ge0$ such that $0<\M_*^D(A)\le\M^{*D}(A)<\ty$, we say that $A$ is {\em Minkowski nondegenerate},\label{nondeg}
and {\em Minkowski degenerate} otherwise.
(Note that if $A$ is nondegenerate, it then follows that $\dim_BA$ exists and is equal to~$D$.)
 If $\M_*^D(A)=\M^{*D}(A)$, their common value is denoted by $\M^D(A)$ and called the {\em Minkowski content} of $A$.\label{minkc} If, in addition, $\M^D(A)\in(0,+\ty)$, then $A$ is said to be {\em Minkowski measurable}.\label{Minkowski_measurable}
 \medskip

We shall need the notion of {\em bounded fractal string} $\mathcal L$.  It is defined as a nonincreasing sequence $\mathcal L=(\g_k)_{k\in\eN}$ of positive real numbers such that $|{\mathcal L}|_1:=\sum_{k=1}^\ty\g_k<\ty$. A {\em tensor product}\label{otimes} ${\mathcal L}_1\otimes\mathcal{L}_2$ of two bounded fractal strings ${\mathcal L}_1$ and ${\mathcal L}_2$ is defined as the bounded fractal string consisting of all possible products $\g\cdot\mu$ with $\g\in {\mathcal L}_1$ and $\mu\in{\mathcal L}_2$, counting the multiplicities. It is clear that $|{\mathcal L}_1\otimes{\mathcal L}_2|_1=|{\mathcal L}_1|_1\cdot|{\mathcal L}_2|_1$. Similarly, we can define the {\em union} of two bounded fractal strings, ${\mathcal L}_1\sqcup{\mathcal L}_2$,\label{sqcup} as the union of the corresponding multisets; that is, as the usual union but also taking the multiplicities into account.

We obviously have that $|{\mathcal L}_1\sqcup{\mathcal L}_2|_1=|{\mathcal L}_1|_1+|{\mathcal L}_2|_1$. If we denote the collection of all bounded fractal strings by ${\mathcal L}_b$, it is easy to see that both $({\mathcal L}_b,\otimes)$ and $({\mathcal L}_b,\sqcup)$ are commutative semigroups, while ${\mathcal L}_b$ is a convex cone in the standard Banach space $(\ell_1(\eR),+)$ of absolutely summable sequences $(x_k)_{k\in\eN}$ of real numbers.

The union $\sqcup$ can be extended to include an infinite sequence of bounded fractal strings $({\mathcal L}_k)_{k\in\eN}$. 
More precisely, the union $\sqcup_{k=1}^\ty{\mathcal L}_k$ is a bounded fractal string, provided $\sum_{k=1}^\ty|{\mathcal L}_k|_1<\ty$. Finally, for any bounded fractal string $\mathcal L$ and $c>0$, we let 
$c\mathcal L:=(c\g)_{\g\in\mathcal L}$.

Throughout the paper, $s$ is a complex variable.
Given a meromorphic function $\varphi=\varphi(s)$ in a neighborhood of $s=\omega\in\Ce$, we denote by $\res(\varphi,\omega)$ its residue at $s=\omega$.
Furthermore, given $\a\in\eR\cup\{\pm\ty\}$, we let, for example, $\{\re s>\a\}$ denote the open right half-plane $\{s\in\Ce:\re s>\a\}$, with the obvious conventions when $\a=\pm\ty$.
In the sequel, all of the bounded subsets $A$ under consideration will be implicitly assumed to be nonempty.
Finally, if $A$ is a nonempty subset of $\eR^N$ and $\g$ a real number, we let $\g A:=\{\g a\in\eR^N:a\in A\}$.

\section{Distance and tube zeta functions}\label{ch_distance}
\subsection{Basic properties of the distance zeta functions of fractal sets}\label{properties}
In this subsection, we recall the definition and some basic properties of the distance zeta functions,
introduced by the first author in 2009, and studied in \cite{dtzf}; see Definition \ref{defn}.
They represent a natural extension of the notion of geometric zeta function of bounded fractal strings.


\begin{defn}\label{defn}  Let $\delta$ be a fixed positive number and let $A$ be a bounded set in $\eR^N$.
The {\em distance zeta function $\zeta_A$} of $A$ is defined by
\begin{equation}\label{z}
\zeta_A(s)=\int_{A_\delta}d(x,A)^{s-N}\D x,
\end{equation}
for all $s\in\Ce$ with $\re s$ sufficiently large. We shall also sometimes write $\zeta_{A,A_\d}(s)$
instead of $\zeta_A(s)$ (as in part $(d)$ of Theorem \ref{an} below), in order to stress that the distance zeta function depends on $\d$ as well.
\end{defn} 

We denote by $D(\zeta_A)$ the {\em abscissa of convergence} of $\zeta_A$ (really, the abscissa of absolute or Lebesgue convergence of $\zeta_A$). It is defined by 
\begin{equation}\label{Dzeta}
D(\zeta_A):=\inf\Bigg\{\a\ge0:\int_{A_\d}d(x,A)^{\a-N}<\ty\Bigg\}.
\end{equation}
Then, 
$\{\re s>D(\zeta_A)\}$, the {\em half-plane of convergence} of $\zeta_A$, is the largest open right half-plane on which the Lebesgue integral defining $\zeta_A$ in Eq.\  \eqref{z} is convergent (or, equivalently, absolutely convergent); see part (b) of Theorem~\ref{an}.

We point out that changing the value of $\d$ amounts to adding an entire function to $\zeta_A=\zeta_A(\,\cdot\,,A_\d)$. Hence, neither the existence of a meromorphic continuation to a domain $U\stq\Ce$, nor the poles of $\zeta_A$ in $U$ and the corresponding residues (or, more generally, principal parts), depend on the choice of $\d$. Exactly the same comment applies to $\tilde\zeta_A$, to be introduced in Definition \ref{zeta_tilde}.

In the sequel, we will need in an essential manner Theorem \ref{an} below, which is the first key result about distance zeta functions and is established in [LapRa\v Zu2].
 A systematic study of distance (and other fractal) zeta functions, containing many additional results, can be found in the monograph \cite{fzf}. In part $(c)$ of Theorem \ref{an} below, $D_{\rm hol}(\zeta_A)$, the {\em abscissa of holomorphic continuation} of $\zeta_A$, is defined exactly as in Definition \ref{Dmer} below, except with ``meromorphic'' replaced by ``holomorphic''. Furthermore, accordingly, $\{\re s>D_{\rm hol}(\zeta_A)\}$ is called the {\em half-plane of holomorphic continuation of $\zeta_A$} and is the largest open right half-plane to which $\zeta_A$ can be holomorphicaly continued.

\begin{theorem} 
\label{an} Let $A$ be an arbitrary bounded subset of $\eR^N$ and let $\delta>0$. Then$:$

\smallskip

$(a)$ The distance zeta function $\zeta_A$ defined by \eqref{z} is holomorphic in the 
half-plane $\{\re s>\ov\dim_BA\}$.

\medskip

$(b)$ The lower bound in the half-plane of convergence $\{\re s>\ov\dim_BA\}$ is optimal, in the sense that $\overline\dim_BA=D(\zeta_A)$, where $D(\zeta_A)$ is the abscissa of Lebesgue $($i.e., absolute$)$ convergence of $\zeta_A$. 


\medskip

$(c)$ If the box $($or Minkowski$)$ dimension $D:=\dim_BA$ exists, $D<N$, and $\M_*^D(A)>0$, then $\zeta_A(s)\to+\ty$ as $s\in\eR$
converges to $D$ from the right. So that, under these hypotheses, $\{\re s>\dim_BA\}$ is also the largest open right half-plane on which $\zeta_A$ is holomorphic; i.e., $\dim_BA=D(\zeta_A)=D_{\rm hol}(\zeta_A)$ and so, the half-planes of convergence and of holomorphic continuation of $\zeta_A$ coincide.

\medskip

$(d)$ For any $\g>0$, we have $D(\zeta_{\g A,\g(A_\d)})=D(\zeta_{A,A_\d})=\ov\dim_BA$ and
\begin{equation}\label{zetalA}
\zeta_{\g A,\g(A_\d)}(s)=\g^s\zeta_{A,A_\d}(s),
\end{equation}
for all $s\in\Ce$ with $\re s>\ov\dim_BA$. Furthermore, if $\o\in\Ce$ is a simple pole of the meromorphic extension of $\zeta_A(s,A_\d)$ to some open connected neighborhood of the critical line $\{\re s=\ov\dim_BA\}$ $($here and thereafter, we use the same notation for the meromophically extended function$)$, then
\begin{equation}\label{reslA}
\res(\zeta_{\g A,\g(A_\d)},\o)=\g^{\o}\res(\zeta_{A,A_\d},\o).
\end{equation}
\end{theorem}

\medskip

In the following definition, we have in mind, especially, meromophic functions in the form of Dirichlet series, or, more generally, defined as Dirichlet-type integrals, in the sense of [LapRa\v Zu1--2].


\begin{defn}\label{Dmer}
Let $f:U\to\Ce$ be a meromorphic function on a domain $U\subseteq\Ce$.
We define the {\em abscissa of meromorphic continuation}\label{a_Mer} $D_{\rm mer}(f)$ 
of $f$ as the infimum of all real numbers $\s$ such that $f$ possesses a meromorphic extension to the open right half-plane $\{\re s>\s\}$. Equivalently, 
$
\operatorname{Mer}(f):=\{\re s>D_{\rm mer}(f)\}
$
is the largest open right
half-plane to which $f$ can be meromorphically extended. We then call $\operatorname{Mer}(f)$ the {\em half-plane of meromorphic continuation} of $f$. By definition, we have $D_{\rm mer}(f)\in\eR\cup\{\pm\ty\}$ and, by convention, $\{\re s>D_{\rm mer}(f)\}$ is equal to $\Ce$ or $\emptyset$ when $D_{\rm mer}(f)=+\ty$ or $-\ty$, respectively.
\end{defn}

Clearly, we always have $D_{\rm mer}(f)\le D(f)$. It is not difficult to show that this inequality is sharp, in general, even for $f=\zeta_A$, $f=\tilde\zeta_A$, or $f=\zeta_{\mathcal L}$, as defined respectively in Definition \ref{defn} above, Definition \ref{zeta_tilde} below, or in \S\ref{zeta_s} just below.

\subsection{Zeta functions of fractal strings and of associated fractal sets}\label{zeta_s}\label{properties_zeta}
The following example shows that Definition \ref{defn} provides a natural extension of the zeta function associated with a (bounded) fractal string $\mathcal L=(\ell_j)_{j\ge1}$, where $(\ell_j)_{j\geq 1}$ is a nonincreasing infinite sequence of positive numbers such that $\sum_{j=1}^\infty \ell_j<\ty$:
\begin{equation}\label{string}
\zeta_{\mathcal L}(s)=\sum_{j=1}^\infty \ell_j^s,
\end{equation}
for all $s\in\Ce$ with $\re s$ sufficiently large. This zeta function $\zeta_{\mathcal L}$, called the {\em geometric zeta function} of $\mathcal L$.
In the sequel, the fractal string $\mathcal L=(\ell_j)_{j=1}^\ty$ is said to be {\em nontrivial} if the sequence $(\ell_j)_{j=1}^\ty$ is infinite; without loss of generality, we may then assume that $\ell_j\searrow0$.
Recall that for a (nontrivial) bounded fractal string $\mathcal{L}$ a key property of the Dirichlet series $\zeta_{\mathcal L}$ is that its abscissa of convergence $D(\zeta_{\mathcal L})$ coincides with the (inner) Minkowski dimension of $\mathcal{L}$.
In particular, the inner Minkowski dimension of $\mathcal{L}$ is defined in terms of its inner $\delta$-neighborhoods.
We refrain here from going into more details about fractal string theory and instead refer the reader to \cite{lapidusfrank12} and the relevant references therein.

\begin{example}\label{L} Let $\mathcal L=(\ell_j)_{j\ge1}$ be a given {\em nontrivial bounded fractal string}; that is, as above, it has finite {\em total length}:
$\sum_{j\ge1}\ell_j<\ty$. 
Let us define the set 
\begin{equation}\label{AL}
A_{\mathcal L}:=\Bigg\{a_k=\sum_{j\ge k}\ell_j:k\ge1\Bigg\}
\end{equation}
 corresponding to $\mathcal L$. Here, $A_{\mathcal L}$ is viewed as a subset of $\eR$.
Using (\ref{z}), it is possible to show by means of a direct computation (see the details in [LapRa\v Zu2]) that the distance zeta function of $A_{\mathcal L}$ is given by
\begin{equation}\label{cantor_string}
\zeta_{A_{\mathcal L}}(s)=s^{-1}2^{1-s}\zeta_{\mathcal L}(s)+2s^{-1}\delta^{s},
\end{equation}
provided that $\delta\ge l_1/2$. The case when $0<\delta<l_1/2$ yields an analogous relation, 
$\zeta_{A_{\mathcal L}}(s)=u(s)\zeta_{\mathcal L}(s)+v(s),$
where again $u(s):=s^{-1}2^{1-s}$, with a simple pole at $s=0$. 
Note that here, $u(s)$ and $v(s)=v(s,\delta)$ are holomorphic functions in the right half-plane $\{\re s>0\}$. 
Hence, since $\zeta_{\mathcal L}$ is holomorphic for $\re s>\ov\dim_BA$, the same relation still holds for a meromorphic extension of $\zeta_A$
within the right half-plane.
\end{example}

It follows from the above discussion that $D(\zeta_{A_{\cal L}})=D(\zeta_{\cal L})$ and the sets of poles of the meromorphic extensions of $\zeta_{A_{\mathcal L}}$ and $\zeta_{\mathcal L}$
to any open right half-plane $\{\re s>c\}$, with $c\ge0$ $($provided one, and therefore both, of the extensions exist$)$, coincide. The poles of $\zeta_{A_{\mathcal{L}}}$ and $\zeta_{\mathcal{L}}$ in such a domain then have the same multiplicities.

\subsection{Complex dimensions of fractal sets}\label{eqzf}

In this subsection, we introduce the key notion of complex dimensions and of principal complex dimensions (see Definition \ref{1.331/2} and Definition \ref{dimc}, respectively), following and adapting \cite[\S1.2.1 and \S5.1]{lapidusfrank12} to our present more general situation. 
In the following definitions, we implicitly assume that the nonempty bounded subset $A$ of $\eR^N$ has the property that
the associated distance zeta function $\zeta_A$ can be extended to a meromorphic function defined on a domain $U\stq\Ce$, 
where $U$ is an open and connected neighborhood of the open half-plane $\{\re s\geq D(\zeta_A)\}$.
%
 (Following the usual conventions, we still denote by $\zeta_A$ the meromorphic continuation of $\zeta_A$ to $U$, which is necessarily unique due to the principle of analytic continuation.)


\begin{defn}\label{dimc}
The {\em set of principal complex dimensions} of $A$, denoted by $\dim_{PC} A$,
is defined as the set of {\em principal poles} $\po_c(\zeta_A)$ of $\zeta_A$; that is, the set of poles of $\zeta_A$ which are located on the critical line $\{\re s=D(\zeta_A)\}$:
\begin{equation}
\dim_{PC} A=\po_c(\zeta_A)=\{\omega\in U:\mbox{$\omega$ is a pole of $\zeta_A$ and $\re \omega=D(\zeta_A)$}\}.
\end{equation}
Clearly, the above set is independent of the choice of the domain $U$.
\end{defn}

As we see, in Definition \ref{dimc}, if $A\subset\eR^N$ is bounded, the singularities of $\zeta_A$ we are interested in are located on the vertical line $\{\re s=\ov\dim_BA\}$. 

Extending the definition of complex dimensions of fractal strings (and other fractals), we also introduce the following natural higher-dimensional generalization in our context.

\begin{defn}\label{1.331/2}
The {\em set of visible complex dimensions} of $A$ {\em with respect to a given domain $U$} (often called, in short, the {\em set of complex dimensions} of $A$ if no ambiguity may arise or if $U=\Ce$), is defined as the set of all the poles of $\zeta_A$ which are located in the domain $U$:
\begin{equation}\label{1.401/2}
\cal P(\zeta_A)=\cal P(\zeta_A,U):=\{\omega\in U : \omega\textrm{ is a pole of } \zeta_A\}.
\end{equation}
\end{defn}

It is easy to check that if the domain $U$ is symmetric with respect to the real axis, then the real complex dimensions come in complex conjugate pairs.
Also, by Theorem \ref{an}$(a)$, $\cal P(\zeta_A)\subseteq\{\re s\leq D(\zeta_A)\}$.

\section{Residues of zeta functions and Minkowski contents}\label{residues_m}

In this section, we recall that some important information concerning the geometry of fractal sets in $\eR^N$ is encoded in their associated fractal zeta functions.
Therefore, the distance zeta functions, as well as the tube zeta functions which we introduce below (see Definition \ref{zeta_tilde}), can be considered as a useful tool in the study of the geometric properties of fractals. 

\subsection{Distance and tube zeta functions of fractal sets, and their residues}\label{residues_m_distance}
The residue of any meromorphic extension of the distance zeta function of a fractal set $A$ in $\eR^N$ is closely
related to the Minkowski content of the set; see Theorem~\ref{pole1}. We use the notation $\zeta_{A,A_\delta}(s)$ for the zeta function
instead of $\zeta_A(s)$, since we want to 
stress the dependence of the zeta function on~$\delta$. We start with a result which is interesting in itself and also leads to a new class of zeta functions, described by Definition \ref{zeta_tilde}; for the proof, see [LapRa\v Zu1--2].

\begin{theorem}\label{equr}
Let $A$ be a bounded set in $\eR^N$, and let $\delta$ be a fixed positive number. Then, for all $s\in\Ce$ such that $\re s>\ov\dim_BA$, 
the following identity holds$:$
\begin{equation}\label{equality}
\int_{A_\delta}d(x,A)^{s-N}\D x=\delta^{s-N}|A_\delta|+(N-s)\int_0^\delta t^{s-N-1}|A_t|\,\D t.
\end{equation}
Furthermore, the Dirichlet-type integral function $\tilde\zeta_A(s):=\int_0^\delta t^{s-N-1}|A_t|\,\D t$ is $($Le\-bes\-gue, i.e., absolutely$)$ convergent 
on $\{\re s>\ov\dim_BA\}$; that is, $D(\tilde\zeta_A)\le\ov\dim_BA$.
\end{theorem}


The following result, also obtained in [LapRa\v Zu2], is important for the development of the theory. Note that it assumes that $A$ is such that the corresponding distance zeta function can be meromorphically extended. On of our goals is to significantly refine Theorem \ref{pole1}; see, especially, Theorem \ref{nonmeasurable} in \S\ref{merom_ext_meromorphic} below.

\begin{theorem}\label{pole1}
Assume that $A\st\eR^N$ is Minkowski nondegenerate; that is, $0<\M_*^D(A)\le\M^{*D}(A)<\ty$ $($in particular, $\dim_BA=D)$. If $D<N$ and $\zeta_A=\zeta_{A,A_\delta}$ can be extended meromorphically to a connected open neighborhood of $s= D$,
then $D$ is necessarily a simple pole of $\zeta_A$, and 
\begin{equation}\label{res}
(N-D)\M_*^D(A)\le\res(\zeta_A,D)\le(N-D)\M^{*D}(A).
\end{equation}
 Furthermore, the value of $\res(\zeta_{A,A_\delta}, D)$ does not depend on $\delta>0$.
In particular, if $A$ is Minkowski measurable, then 
\begin{equation}\label{pole1minkg1=}
\res(\zeta_A, D)=(N-D)\,\M^D(A).
\end{equation}
\end{theorem}

In light of Theorem~\ref{equr}, it is natural to introduce a new fractal zeta function of bounded sets $A$ in $\eR^N$.

\begin{defn}\label{zeta_tilde}  Let $\delta$ be a fixed positive number, and let $A$ be a bounded set in $\eR^N$. Then, the {\em tube zeta function} of $A$, denoted by $\tilde\zeta_A$, is defined (for $\re s$ large enough) by
\begin{equation}\label{tildz}
\tilde\zeta_A(s)=\int_0^\delta t^{s-N-1}|A_t|\,\D t.
\end{equation}

As we know from Theorem \ref{equr}, we have $D(\tilde\zeta_A)\le\ov\dim_BA$. In particular, since $D_{\rm hol}(\tilde\zeta_A)\le D(\zeta_A)$, we conclude that $D_{\rm hol}(\zeta_A)\le\ov\dim_BA$; that is, the tube zeta function is holomorphic on $\{\re s>\ov\dim_BA\}$.
\end{defn} 

We call $\tilde\zeta_A$ the tube zeta function of $A$ since its definition involves the tube function $(0,\delta)\ni t\mapsto |A_t|$. Relation (\ref{equality}) can be written as follows, for any $\d>0$ and for $\re s>\ov{\dim}_BA$:\footnote{Note that both, $\zeta_A$ and $\tilde\zeta_A$ depend on $\d$; see the comment following Eq.\ \eqref{Dzeta}.}
\begin{equation}\label{equ_tilde}
\zeta_{A}(s)=\delta^{s-N}|A_\delta|+(N-s)\tilde\zeta_A(s).
\end{equation}

Combining Theorem \ref{an} with Theorem \ref{equr} and using Eq.\  \eqref{equ_tilde}, we can derive the following useful result.

\begin{cor}\label{zetac}
Let $A$ be a bounded subset of $\eR^N$.
Then $($with the notation of \S\ref{properties}$)$,

\begin{equation}\label{Dles}
D_{\rm mer}(\zeta_A)\le D_{\rm hol}(\zeta_A)\le D(\zeta_A)=\ov\dim_BA.
\end{equation}
\smallskip

If, in addition, we assume that $\ov{\dim}_BA<N$, then

\begin{equation}
\begin{gathered}
D(\zeta_A)=D(\tilde\zeta_A)=\ov\dim_BA,\\
D_{\rm hol}(\zeta_A)=D_{\rm hol}(\tilde\zeta_A),\q D_{\rm mer}(\zeta_A)=D_{\rm mer}(\tilde\zeta_A),
\end{gathered}
\end{equation}
\end{cor}
\medskip


\begin{remark}\label{tilde_analog}
The functional equation \eqref{equ_tilde} and the principle of analytic continuation are crucial in deriving many results about the tube zeta function from the distance zeta function, and vice versa.
For instance, the exact analogs of Theorems \ref{an} and \ref{pole1} are valid for the tube zeta function $\tilde{\zeta}_A$ of a given bounded subset $A$ in $\eR^N$.
The only difference is that in the case of the analog of Theorem \ref{pole1} for the tube zeta function, the multiplicative constant $(N-D)$ is absent from the analogs of \eqref{res} and \eqref{pole1minkg1=}, and the condition $D<N$ can be omitted.
\end{remark}


\begin{remark}\label{1.3.121/2}
Assume that $\ov D=\ov{\dim}_BA<N$. In light of Eq.\ ~\eqref{equ_tilde}, for any $\delta>0$, the two zeta functions $\tilde{\zeta}_A(s):=\tilde\zeta_A(s,\d)$ and $\zeta_A(s):=\zeta_{A,A_\delta}(s)$ differ by an entire function. Therefore, if one of them has a meromorphic continuation to some open connected set $U\supseteq\{\re s=\ov D\}$, so does the other, and then they have exactly the same poles in $U$ (with the same multiplicities): $\po(\tilde{\zeta}_A)=\po(\zeta_A)$ or, more precisely, $\po(\tilde{\zeta}_A,U)=\po(\zeta_A,U)$. In particular, $\po_c(\tilde{\zeta}_A)=\po_c(\zeta_A)$.
Furthermore, in that case, if $\omega\in U$ is a simple pole of $\tilde{\zeta}_A(s)$, then it is also a simple pole of $\zeta_A(s)$ and the corresponding residues are linked as follows:
\begin{equation}\label{1.3.181/2}
\res(\tilde{\zeta}_A,\omega)=\frac{1}{N-\omega}\res(\zeta_A,\o).
\end{equation}
Note that we must have $\omega\neq N$ since $\re\omega\leq\ov D$ and $\ov D<N$.
\end{remark}

\begin{example}[$a$-strings]\label{a-string2}
Given $a>0$, let $A=\{j^{-a}:j\in\eN\}$. This set is Minkowski measurable,
\begin{equation}\label{a-string}
\M^D(A)=\frac{2^{1-D}}{D(1-D)}a^D,\quad D=D(a)=\frac1{1+a},
\end{equation} 
and the related string $\mathcal L=(\ell_j)_{j\ge1}$ defined by $\ell_j=j^{-a}-(j+1)^{-a}$ is called the {\em $a$-string}; see \cite{Lap1}, \cite{lapiduspom} and \cite[\S6.5.1]{lapidusfrank12} for the study of its various properties. Due to (\ref{pole1minkg1=}) and (\ref{1.3.181/2}), we know that
\begin{equation}\label{a-string-mink}
\res(\zeta_A,D)=(1-D)\M^D(A),\quad \res(\tilde\zeta_A,D)=\M^D(A).
\end{equation} 
The last equality is not fortuitous. In fact, it holds for a very general class of Minkowski measurable sets; see Theorem \ref{measurable} in \S\ref{merom_ext_meromorphic} just below.
\end{example}

\subsection{Distance zeta function of the Sierpi\'nski carpet}\label{scarpet}
 In this subsection, we compute the principal complex dimensions of the Sierpi\'nski carpet. In order to do this, we must first describe the computation of its distance zeta function $\zeta_A$.
The following proposition is the main result of \S\ref{scarpet}.

\begin{prop}\label{sierpinski_carpet0}
Let $A$ be the Sierpi\'nski carpet in $\eR^2$, constructed in the usual way inside the unit square, and let $\d>0$ be fixed. We assume without loss of generality that $\d>1/6$, so that the set $A_\d$ is connected.
Then, $\zeta_A$ ha a meromorphic continuation to all of $\Ce$ given by 
\begin{equation}\label{zeta_carpet}
\zeta_A(s)=\frac{8}{2^ss(s-1)(3^s-8)}+2\pi\frac{\d^s}s+4\frac{\d^{s-1}}{s-1}.
\end{equation}
Hence, the set of principal complex dimensions of $A$ is given by
\begin{equation}\label{princ}
\dim_{PC} A= \log_38+\frac{2\pi}{\log 3}\I\Ze
\end{equation}
and consists only of simple poles of $\zeta_A$. Also, the residues of $\zeta_A$ computed
at the principal poles $s_k:=\log_38+\frac{2\pi}{\log 3}k\I\in\dim_{PC} A$, with $k\in\Ze$, are given by
$$
\res(\zeta_A,s_k)=\frac{2^{-s_k}}{(\log3)s_k(s_k-1)}.
$$
\end{prop}

As we see from Eq.\  \eqref{zeta_carpet}, the set of complex dimensions (i.e., the set of poles of $\zeta_A$ in all of $\Ce$) of $A$ consists only of simple poles of $\zeta_A$ and is given by
\begin{equation}
\po(\zeta_A)=\{0,\,1\}\cup\Big(\log_38+\frac{2\pi}{\log 3}\I\Ze\Big).
\end{equation}

For the needs of the proof of Proposition \ref{sierpinski_carpet0}, it will be very convenient here to introduce some auxilliary notation. Let $A$ be a compact subset of $\eR^2$ and assume that $\O$ is a bounded open (or more generally, a bounded and Lebesgue measurable) subset of $\eR^2$. Then we define (for $\re s$ large enough)
\begin{equation}
\zeta_{A,\O}(s):=\int_{\O}d((x,y),A)^{s-2}\D x\,\D y,
\end{equation}
We call it the {\em relative distance zeta function} of $A$ with respect to $\O$. 
Such distance zeta functions and their generalizations, associated with suitable ordered pairs $(A,\O)$ of subsets of $\eR^N$ (called {\em relative fractal drums} or RFDs in $\eR^N$), are studied in detail in [LapRa\v Zu3].

\begin{proof}[Proof of Proposition \ref{sierpinski_carpet0}]
In order to evaluate 
$
\zeta_A(s)=\int_{A_\d} d((x,y),A)^{s-2}\D x\,\D y,
$ 
we integrate $(i)$~first over the set $A_\d\setminus [0,1]^2$, and then $(ii)$~over the unit square $[0,1]^2$. 

\bigskip

{\em Step} ($i$): The integration over the set $A_\d\setminus (0,1)^2$ leads us to the following result:
\begin{equation}
\begin{aligned}
\zeta_{A,A_\d\setminus[0,1]^2}(s)&=\int_0^{2\pi}\D\f \int_0^\d r^{s-2}r\,\D r+4\int_0^1\D x\int_0^\d y^{s-2}\D y\\
&=2\pi\frac{\d^s}s+4\frac{\d^{s-1}}{s-1},
\end{aligned}
\end{equation}
for $\re s>1$.
Indeed, it suffices to note that the (connected) set $A_\d\setminus [0,1]^2$ can be viewed as the disjoint union of 

\medskip

$\bullet$ four quarters of the corresponding disks of radius $\d$, with centers at the vertices of the unit square $[0,1]^2$  
and 
\smallskip

$\bullet$ of the remaining four rectangles, which are all isometrically isomorphic to $[0,1]\times(0,\d)$.

\bigskip

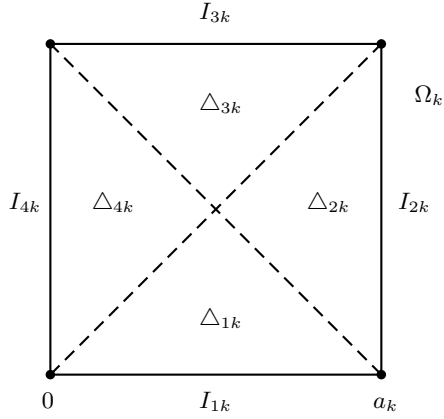
\begin{figure}[t]
\psset{unit=1.1}
\begin{pspicture}(-3.5,-0.5)(3.6,4.5)

\pstGeonode[PointName=none](-2,0){A}(2,0){B}(2,4){C}(-2,4){D}
\pstLineAB{A}{B}
\pstLineAB{B}{C}
\pstLineAB{C}{D}
\pstLineAB{D}{A}
\put(-2.1,-0.4){\small$0$}
\put(1.9,-0.4){\small$a_k$}
\put(2.4,3.3){\small$\O_k$}
\put(-0.2,-0.4){\small$I_{1k}$}
\put(-0.2,0.6){\small$\triangle_{1k}$}
\put(-0.2,4.3){\small$I_{3k}$}
\put(-0.2,3.2){\small$\triangle_{3k}$}
\put(-2.5,2){\small$I_{4k}$}
\put(-1.5,2){\small$\triangle_{4k}$}
\put(2.2,2){\small$I_{2k}$}
\put(1.1,2){\small$\triangle_{2k}$}

\pstLineAB[linestyle=dashed]{A}{C}
\pstLineAB[linestyle=dashed]{B}{D}

\end{pspicture}
\caption{The square $\O_k$ corresponds to any of the $8^{k-1}$ deleted open square in the $k$-th generation during the construction of the Sierpi\'nski carpet. It can be viewed as the union of four triangles, determined by its diagonals.
This figure explains a part of Step ($ii$) in the proof of Proposi\-ti\-on~\ref{sierpinski_carpet0}.}
\label{sierpinski_carpet1k}
\end{figure}

{\em Step} ($ii$): Now, let us consider $\zeta_{A,[0,1]^2}$. Since the boundary of the unit square $[0,1]^2$ is of $2$-dimensional Lebesgue measure zero, it suffices to consider $\zeta_{A,(0,1)^2}(s)$. Furthermore, since $|A|_2=0$, it suffices to consider
\begin{equation}\label{zeta_sierpinski_carpet}
\zeta_{A,[0,1]^2}(s)=\zeta_{A,(0,1)^2\setminus A}(s)=\sum_{k=1}^\ty 8^{k-1}\zeta_{A_k,\O_k}(s),
\end{equation}
where $\O_k$ is a fixed deleted square of sides of length $a_k:=3^{-k}$ in the $k$-th generation, $A_k=\pa\O_k$ is the boundary of $\O_k$ and $k$ is any positive integer.
Recall that the $k$-th generation of deleted squares contains precisely $8^{k-1}$ deleted squares which are all isometric to $\O_k$. 
Hence. it follows that all of the distance zeta functions corresponding to the deleted open squares $\O_k$ in the $k$-th generation coincide.   Now, if we denote the side length of $\O_k$ by $a_k:=3^{-k}$, it is easy to check that
\begin{equation}\label{8zeta}
\zeta_{A_k,\O_k}(s)=\frac{8\cdot 2^{-s}a_k^{s}}{s(s-1)},
\end{equation}
for $\re s>1$. 
It then follows from the principle of analytic continuation that \eqref{8zeta} continues to hold for all $s\in\Ce$.
Indeed, for any of the four sides $I_{1k}$, $I_{2k}$, $I_{3k}$ and $I_{4k}$ of the square $\O_k$, it is natural to consider the set of points $(x,y)\in\O_k$ such that 
$
d((x,y),\pa\O_k)=d((x,y),I_{ik}),$ 
for
$
i=1,2,3,4.
$
It is easy to see that this set is a triangle, and therefore we can decompose $\O_k$ into the union of four isosceles right triangles $\triangle_{ik}$, $i=1,2,3,4$ (each of them corresponding to one of the four sides of the square $\O_k$), as indicated in Figure \ref{sierpinski_carpet1k}. Note that the triangles are determined by the two diagonals of $\O_k$.
Obviously, 
$
\zeta_{A_k,\O_k}(s)=4\zeta_{I_{1k},\triangle_{1k}}(s),
$
while an easy computation in local Cartesian coordinates yields that
$
\zeta_{I_{1k},\triangle_{1k}}(s)={2(a_k/2)^{s}}/{(s(s-1))},
$
for $\re s>1$.
Therefore, we obtain \eqref{8zeta}. Substituting \eqref{8zeta} into \eqref{zeta_sierpinski_carpet}, we 
conclude that
$$
\zeta_{A,[0,1]^2}(s)=\frac{2^{-s}}{s(s-1)}\sum_{k=1}^\ty 8^{k}3^{-ks}=\frac{8}{2^ss(s-1)(3^s-8)},
$$
for $\re s>\log_28$. Naturally, by analytic continuation, this same equation continues to hold for all $s\in\Ce$.
\medskip

{\rm Step} ($iii$): The resulting expression for $\zeta_A$ stated in Eq.\  \eqref{zeta_carpet} follows from Steps ($i$) and ($ii$).
By the principle of analytic continuation, $\zeta_A$ can be meromorphically extended to all of $\Ce$ and is given by the same formula. Hence, the principal complex dimensions $s_k\in\dim_{PC} A$, $k\in\Ze$, are given by Eq.\ \eqref{princ}. Finally, we omit the easy computation of the residues $\res(\zeta_A,s_k)$.
\end{proof}

\bigskip

\section{Meromorphic extensions of distance and tube zeta functions}\label{merom_ext}

The aim of this section is to describe a construction of meromorphic extensions of distance and tube zeta functions, and obtain several refinements of Theorem~\ref{pole1}  
and its counterpart for the tube zeta function. The main results of this section are stated in Theorems~
\ref{measurable} and~\ref{nonmeasurable}, dealing with tube zeta functions of Minkowski measurable and Minkowski nonmeasurable sets, respectively.

\subsection{Meromorphic extensions of tube and distance zeta functions}\label{merom_ext_meromorphic}
The following theorem shows that the tube zeta function of a large class of Minkowski measurable sets possess a nontrivial meromorphic extension,
assuming a mild technical condition on the growth rate of the tube function $t\mapsto|A_t|$.
As we see from Theorems~\ref{measurable} and~\ref{nonmeasurable} below, the second term in the asymptotic expansion of the tube function plays a crucial role in order to reach such a conclusion. 
By using the functional equation \eqref{equ_tilde} connecting $\zeta_A$ and $\tilde\zeta_A$, it is easy to derive
the corresponding counterparts of these results for the distance zeta functions of bounded sets and for the geometric zeta functions of bounded fractal strings; see Theorem \ref{meas_nonmeas} below, which extends to $\zeta_A$ both Theorem \ref{measurable} and Theorem \ref{nonmeasurable}.

\begin{theorem}[Minkowski measurable case]\label{measurable}
Let $A$ be a subset of $\eR^N$ such that there exist  $\alpha>0$, $\mathcal M\in(0,+\ty)$ and $D\ge0$ satisfying
\begin{equation}\label{A_t}
|A_t|= t^{N-D}\left({\mathcal M}+O(t^\alpha)\right)\quad\mathrm{as}\quad t\to0^+.
\end{equation}
Then, $\dim_BA=D$ and $A$ is Minkowski measurable with $\mathcal M^D(A)=\cal M$. Furthermore, the tube zeta function $\tilde\zeta_A$ has for abscissa of convergence $D(\tilde\zeta_A)=D$ and possesses a unique meromorphic continuation $($still denoted by $\tilde\zeta_A)$ to $($at least$)$ the open right half-plane $\{\re s>D-\alpha\}$; that is,
\begin{equation}
D_{\rm mer}(\tilde\zeta_A)\le D-\a.
\end{equation}
The only pole of $\tilde\zeta_A$ in this half-plane is $s=D$; it is simple, and $\res(\tilde\zeta_A,D)=\M$.
\end{theorem}

\begin{proof}
We have
\begin{equation}\nonumber
\begin{aligned}
\tilde\zeta_A(s)&=\int_0^\delta t^{s-N-1}|A_t|\,\D t=\int_0^\delta t^{s-N-1}t^{N-D}({\mathcal M}+O(t^\alpha))\,\D t\\
&=\underbrace{{\mathcal M}\frac{\delta^{s-D}}{s-D}}_{\zeta_1(s)}+\underbrace{\int_0^\delta t^sO(t^{-D+\alpha-1})\,\D t}_{\zeta_2(s)},
\end{aligned}
\end{equation}
provided $s\in\Ce$ is such that $\re s>D$.
The function $\zeta_1(s)$ is meromorphic in the entire complex plane and $D(\zeta_1)=D$, while for $\zeta_2(s)$ we have
\begin{equation}\nonumber
|\zeta_2(s)|\le K\int_0^\delta t^{\re s-D+\alpha-1}\D t<\ty
\end{equation}
for $\re s>D-\alpha$, where $K$ is a positive constant. Therefore, $D(\zeta_2)\le D-\alpha<D=D(\zeta_1)$, and the claim now follows by applying the principle of meromorphic continuation to the relation $\tilde\zeta_A(s)=\zeta_1(s)+\zeta_2(s)$ (valid for all $s\in\Ce$ such that $\re s>D$), by observing that the right-hand side possesses a (necessarily unique) meromorphic extension to the open right half-plane $\{\re s>D-\a\}$.
\end{proof}

It can be shown that, in some precise sense, the sufficient condition \eqref{A_t} for meromorphic extendibility of the tube zeta function, appearing in Theorem \ref{measurable}, is also necessary; see [LapRa\v Zu8, Thm.\ 5.7]. 

\begin{remark}\label{ext_measurable}
The function of order $O(t^\alpha)$ as $t\to0^+$, appearing in Theorem~\ref{measurable} or in Theorem \ref{nonmeasurable} below, can be replaced by a function of order $O(t^{\alpha)})$ as $t\to0^+$, in the precise sense of Definition~\ref{4.181/2}.
These more general functions include, for example, $t^\alpha\log(1/t)$ or 
$t^\alpha\log\ldots\log(1/t)$,
near $t=0^+$, where $\a>0$ and the last factor is the $q$-th iterated logarithm for any integer $q\geq 1$.
\end{remark}

\begin{defn}\label{4.181/2}
Let $f$ be defined on an interval $(0,\delta)$, for some $\delta>0$. Then, given $\alpha\in\eR$, $f$ is said to be {\em of order} $O(t^{\alpha)})$ {\em as} $t\to 0^+$ (which we write $f(t)=O(t^{\alpha)})$ as $t\to 0^+$) if for every $\alpha_0<\alpha$, it is of order $O(t^{\alpha_0})$ as $t\to 0^+$.
\end{defn}

The following theorem deals with a class of bounded sets in $\eR^N$ that are not Minkowski measurable.
More specifically, we deal with the sets $A$ such that $0\le\M_*^D(A)<\M^{*D}(A)<\ty$, where $D=\dim_BA$.

Let us first introduce some notation.
Given a $T$-periodic function $G:\eR\to\eR$, we denote by $G_0$ its truncation to $[0,T]$; that is,
\begin{equation}
G_0(\tau)=
\begin{cases}
G(t)& \mbox{if $\tau\in[0,T]$}\\
0& \mbox{if $\tau\notin[0,T]$}.
\end{cases}
\end{equation}
Furthermore, the Fourier transform of $G_0$
is denoted by $\hat G_0$:
\begin{equation}\label{fourier}
\hat G_0(t)=\int_{-\ty}^{+\ty}\E^{-2\pi \I \,t\tau}G_0(\tau)\,\D\tau=\int_0^T\E^{-2\pi \I \,t\tau}G(\tau)\,\D\tau.
\end{equation}
 
\medskip

\begin{theorem}[Minkowski nonmeasurable case]\label{nonmeasurable} 
Let $A$ be a bounded subset of $\eR^N$ such that there exist $D\ge0$, $\alpha>0$, and $G:\eR\to[0,+\ty)$ a nonconstant periodic function with minimal period $T>0$, satisfying
\begin{equation}\label{G}
|A_t|=t^{N-D}\left(G(\log t^{-1})+O(t^\alpha)\right)\quad\mbox{\rm as\q$t\to0^+$.}
\end{equation}
 Then $\dim_BA=D$, $G$ is continuous, and $A$ is Minkowski nondegenerate with upper and lower Minkowski contents respectively given by
\begin{equation}\label{1.4.201/2}
\M_*^D(A)=\min G,\quad \M^{*D}(A)=\max G.
\end{equation}
$($Hence, the range of $G|_{[0,T]}$ is equal to $[\M_*^D(A),\M^{*D}(A)]$.$)$ 
Furthermore, the tube zeta function $\tilde\zeta_A$ has for abscissa of convergence $D(\tilde\zeta_A)=D$ and possesses a unique meromorphic extension $($still denoted by $\tilde\zeta_A$$)$
to {\rm({\it at least})} the open right half-plane $\{\re s>D-\alpha\}$; that is,
\begin{equation}
D_{\rm mer}(\tilde\zeta_A)\le D-\a.
\end{equation}
Moreover, the set of all the poles of $\tilde\zeta_A$ located in this half-plane is given by
\begin{equation}\label{Dpoles}
\mathcal P(\tilde \zeta_A)=\left\{s_k=D+\frac{2\pi}T\I \,k:\hat G_0\Big(\frac kT\Big)\ne0,\,\,k\in\Ze\right\}
\end{equation} 
and this set coincides with the set of principal complex dimensions of $A$; that is, with $\dim_{PC}A:=\po_c(\tilde{\zeta}_A)$, in the notation of Definition~\ref{dimc}. 
The complex dimensions are all simple, and the residue at each $s_k\in\mathcal P(\tilde\zeta_A)$, $k\in\Ze$, is given by
\begin{equation}\label{res_fourier}
\res(\tilde\zeta_A,s_k)=\frac1T\hat G_0\Big(\frac kT\Big).
\end{equation}
If $s_k\in \mathcal P(\tilde \zeta_A)$, then $s_{-k}\in \mathcal P(\tilde \zeta_A)$, and
\begin{equation}\label{riemann_lebesgue}
|\res(\tilde\zeta_A,s_k)|\le \frac1T\int_0^TG(\tau)\,\D\tau,\quad \lim_{k\to\ty}\res(\tilde\zeta_A,s_k)=0.
\end{equation}
In addition, the set of poles $\mathcal P(\tilde\zeta_A)$ $($i.e., of complex dimensions of $A)$ contains $s_0=D$, and
\begin{equation}\label{avarage}
\res(\tilde\zeta_A,D)=\frac1T\int_0^TG(\tau)\,\D\tau.
\end{equation}
In particular, $A$ is {\rm not} Minkowski measurable and
\begin{equation}\label{res_inequalities}
\M_*^D(A)<\res(\tilde\zeta_A,D)<\M^{*D}(A)<\infty.
\end{equation}
\end{theorem}
\bigskip

\begin{remark}
It can be shown, much as in the proof of \cite[Thm.\ 8.30]{lapidusfrank12},
that under the hypotheses of Theorem \ref{nonmeasurable}, the residue of $\tilde\zeta_A$ at $s=D$
in Eq.\  \eqref{avarage} also coincides with the {\em average Minkowski content} of $A$ (defined as a suitable
Ces\`aro multiplicative average of $t\mapsto t^{-(N-D)}|A_t|$, exactly as in \cite[Def.\ 8.29]{lapidusfrank12}, except with $1$ replaced by $N$). A similar comment applies to the distance zeta function $\zeta_A$ in Theorem \ref{meas_nonmeas} below, with the obvious modifications; compare Eq.\  \eqref{avarage} above and Eq.\  \eqref{avarage1} below.
\end{remark}

In the proof of Theorem~\ref{nonmeasurable}, we shall need the following simple lemma, the proof of which is omitted.

\begin{lemma}\label{osc}
Let $F:(0,\delta)\to\eR$ be continuous, and assume that $G:\eR\to\eR$ is a $T$-periodic function, for some $T>0$.
If $F(t)=G(\log t^{-1})+o(1)$ as $t\to0^+$, then $G$ is continuous.
\end{lemma}

We are now ready to establish Theorem~\ref{nonmeasurable}.

\begin{proof}[Proof of Theorem~\ref{nonmeasurable}] To show that $G$ is continuous, it suffices to apply 
Lemma~\ref{osc} to $F(t):=|A_t|t^{D-N}$, which is defined and continuous for $t>0$.
We can write $\tilde\zeta_A(s)=\zeta_1(s)+\zeta_2(s)$, where
\begin{equation}\label{z1s}
\zeta_1(s)=\int_0^\delta t^{s-D-1}G(\log t^{-1})\,\D t,\quad \zeta_2(s)=\int_0^\delta t^sO(t^{-D+\alpha-1})\,\D t,
\end{equation}
for some $\delta>0$ fixed.
As in the proof of Theorem~\ref{measurable}, we have $D(\zeta_2)=D-\alpha$. Therefore, it suffices to prove that $\zeta_1(s)$ can be meromorphically extended to the whole complex plane. We will show this by computing $\zeta_1(s)$ in a closed form.
Since $G$ is $T$-periodic, we have
$
\zeta_1(s)=\int_0^\delta t^{s-D-1}G(\log t^{-1}+T)\,\D t.
$
Introducing a new variable $u$ defined by $\log u^{-1}=\log t^{-1}+T$, that is, $u=\E^{-T}t$, we obtain
\begin{equation}\nonumber
\begin{aligned}
\zeta_1(s)&=\E^{T(s-D)}\int_0^{\delta \E^{-T}}u^{s-D-1}G(\log u^{-1})\D u=\E^{T(s-D)}\left(\int_0^\delta+\int_\delta^{\E^{-T}\delta}\right)\\
&=\E^{T(s-D)}\left(\zeta_1(s)+\int_\delta^{\E^{-T}\delta}t^{s-D-1}G(\log t^{-1})\,\D t\right).
\end{aligned}
\end{equation}
From this, we immediately obtain $\zeta_1(s)$ in the following closed form:
\begin{equation}\label{z1}
\begin{aligned}
\zeta_1(s)&=\frac{\E^{T(s-D)}}{\E^{T(s-D)}-1}\int_{\E^{-T}\delta}^\delta t^{s-D-1}G(\log t^{-1})\,\D t\\
&=\frac{\E^{T(s-D)}}{\E^{T(s-D)}-1}\underbrace{\int_{\log \delta^{-1}}^{\log\delta^{-1}+T}\E^{-\tau(s-D)}G(\tau)\,\D\tau}_{I(s)},
\end{aligned}
\end{equation}
where in the last equality we passed to the new variable $\tau=\log t^{-1}$.
The last integral $I(s)$ is obviously an entire function of $s$, since $\delta$ is different from~$0$ and~$\ty$. 
This shows that the function $\zeta_1(s)$ is meromorphic on all of $\Ce$, and the set of its poles is equal to the set of complex solutions $s_k$ of $\exp(T(s-D))=1$
for which $I(s_k)\ne0$.
If $I(s_k)=0$, it is easy to see that $s_k$ is a removable singularity of $\zeta_1(s)$:
\begin{equation}\nonumber
\lim_{s\to s_k}\zeta_1(s)=\lim_{s\to s_k}\frac{s-s_k}{\E^{T(s-D)}-1} \E^{T(s-s_k)}\frac{I(s)}{s-s_k}=\frac1T I'(s_k),
\end{equation}
where $I'$ denotes the derivative of $I$.
Since
\begin{equation}\label{Isk}
\begin{aligned}
I(s_k)&=\int_{\log \delta^{-1}}^{\log\delta^{-1}+T} \E^{-2\pi \I\frac kT\cdot\tau}G(\tau)\,\D\tau=\int_0^T \E^{-2\pi \I\frac kT\cdot\tau}G(\tau)\,\D\tau=\hat G_0\Big(\frac kT\Big),
\end{aligned}
\end{equation}
where we have used the fact that both $\tau\mapsto \E^{2\pi \I\frac kT\cdot\tau}$ and $\tau\mapsto G(\tau)$ are $T$-periodic functions,
we conclude that the set of poles of $\tilde\zeta_A$ is described by~(\ref{Dpoles}). Note that it contains $D$, since for $k=0$ we have
$
I(D)=I(s_0)=\hat G_0(0)=\int_0^TG(\tau)\,\D\tau>0.
$
Indeed, the range of the function $G|_{[0,T]}$ is equal to the interval $[\M_*,\M^{*}]$, where 
$\M_*=\M_*^D(A)$ and $\M^{*}=\M^{*D}(A)$.
Since $G$ is assumed to be nonconstant, we deduce from (\ref{G}) that $0\le \M_*<\M^{*}<\ty$.

Therefore, we have $D(\zeta_1)=D>D-\alpha=D(\zeta_2)$, and, since $\tilde\zeta_A(s)=\zeta_1(s)+\zeta_2(s)$ for $\re s>D$, we also see that $\tilde\zeta_A$ possesses a (necessarily unique) meromorphic extension to the open right half-plane $\{\re s>D-\alpha\}$.

Next, we compute the residue of $\tilde\zeta_A$ at $s_k=D+\frac{2\pi}Tk\I \in \mathcal P(\tilde\zeta_A)$ by using l'Hospital's rule and~(\ref{Isk}):
\begin{equation}\nonumber
\begin{aligned}
\res (\tilde\zeta_A,s_k)&=\res (\zeta_1,s_k)=\lim_{s\to s_k}\frac{s-s_k}{\E^{T(s-D)}-1} \,\E^{T(s_k-D)} I(s_k)
=\frac1T\hat G_0\Big(\frac kT\Big).
\end{aligned}
\end{equation}
Substituting $k=0$, we obtain~(\ref{avarage}).
The inequalities in (\ref{res_inequalities}) follow from (\ref{avarage}).

As is well known, since $G_0\in L^1(\eR)$, we have $|\hat G_0(\tau)|\le\|G_0\|_{L^1(\eR)}=\|G\|_{L^1(0,T)}$ and $\lim_{|t|\to\ty}\hat G_0(t)=0$ (by the Riemann--Lebesgue lemma); hence, (\ref{riemann_lebesgue}) follows immediately.
\end{proof}

\begin{example}\label{cantor_nonmeasurable}
Let $A$ be the classic ternary Cantor set in $[0,1]$. According to \cite[Eq.~(1.11)]{lapidusfrank12} (which can also be deduced from the general pointwise tube formula in \cite{mm}), we have
\begin{equation}\label{s-cantor}
|A_t|=t^{1-D}2^{1-D}\left(2^{-\{\log_3(2t)^{-1}\}}+\left(3/2\right)^{\{\log_3(2t)^{-1}\}}\right)
\end{equation}
for all $t\in(0,1/2)$. Here, for $x\in\eR$, $\{x\}:=x-\lfloor x\rfloor\in[0,1)$\label{fractional} denotes the fractional part of $x$, where $\lfloor x\rfloor$ is the integer part or the `floor' of $x$.
Then the condition (\ref{G}) in Theorem~\ref{nonmeasurable} is satisfied for $N=1$, $D=\log_32$, $\alpha=D$, and
\begin{equation}\label{cantorG}
G(\tau):=2^{1-D}\Big(2^{-\left\{\frac{\tau-\log2}{\log3}\right\}}+
(3/2)^{\left\{\frac{\tau-\log2}{\log3}\right\}}\Big).
\end{equation}
It is easy to see that the function $G$ is periodic, with minimal period $T=\log3$, and is continuous. However, it is not of class $C^1$
since it is nondifferentiable at the points $\tau_k=\log2+Tk$, $k\in\Ze$; see \cite[Fig.\ 1.5]{lapidusfrank12}. It is therefore
 convenient to consider the restriction $G|_I$ of $G$
to the interval $I=[\log2,\log2+T]$, since the value of $\M^{*D}(A)$ is achieved at the endpoints of~$I$, and $G|_I$ is convex. 
An easy geometric analysis shows that $\M^{*D}(A)=2^{2-D}\approx2.583$, while the minimum value of $G$ is $\M_*^D(A)=2^{1-D}D^D(1-D)^{-(1-D)}\approx 2.495$, achieved at the minimum of $G|_I$, which is easy to compute. (See also \cite[Thm.\ 4.6]{lapiduspom} or \cite[\S1.1.2]{lapidusfrank12})
According to Theorem~\ref{nonmeasurable}, $\tilde\zeta_A$ has for abscissa of convergence $D(\tilde\zeta_A)=\log_32$; further, it can be meromorphically extended to the open right half-plane $\{\re s>\alpha\}$ for any $\alpha>0$, and hence, to the entire complex plane.
Also, 
\begin{equation}
\mathcal P(\tilde\zeta_A)=\Big\{s_k=D+\frac{2\pi}{\log3}k\I :k\in\Ze\Big\}\cup\{0\}=\Big(D+\frac{2\pi}{\log 3}\I\Ze\Big)\cup\{0\},
\end{equation}
(each being simple), and hence, the set of complex dimensions of the Cantor string (see Remark \ref{s=0} below) is given by
\begin{equation}
\dim_{PC}A=\Big\{s_k=D+\frac{2\pi}{\log3}k\I :k\in\Ze\Big\}=D+\frac{2\pi}{\log 3}\I\Ze,
\end{equation}
in agreement with \cite[Eq.\ (1.30)]{lapidusfrank12}. 
Computing the Fourier transform of $G_0$ directly, with its support shifted to $[\log2,\log2+T]$ (since, on this interval,
formula (\ref{cantorG}) holds without curly brackets), we obtain that for each pole $s_k\in\mathcal P(\tilde\zeta_A)$:
\begin{equation}\nonumber
\begin{aligned}
\res(\tilde\zeta_A,s_k)&=\frac1T\hat G_0\Big(\frac kT\Big)=\frac1T\int_{\log2}^{\log6}\E^{-2\pi i\frac kT\tau}G(\tau)\,\D\tau\\
&=\frac{2^{1-D}}T\Big(\frac{2^D}{s_k}(2^{-s_k}-6^{-s_k})+\frac{(1.5)^{-D}}{1-s_k}(6^{1-s_k}-2^{1-s_k})\Big)=\frac{2^{-s_k}}{Ts_k(1-s_k)},
\end{aligned}
\end{equation}
because $3^{s_k}=2$.
Since $|2^{-s_k}|=2^D$, we conclude that $\res(\tilde\zeta_A,s_k)\asymp k^{-2}$ as $k\to\ty$, which agrees with the limit in (\ref{riemann_lebesgue}).

The residues of $\zeta_A$ and of $\zeta_{\mathcal{L}}$ are obtained by using (\ref{equ_tilde}) (with $N=1$) and (\ref{cantor_string}), respectively:
$$
\res(\zeta_A,s_k)=(1-s_k)\res(\tilde\zeta_A,s_k)=\frac{2^{-s_k}}{Ts_k},\ \res(\zeta_{\mathcal L},s_k)=s_k2^{s_k-1}\res(\zeta_A,s_k)=\frac 1{2T}.
$$
For the Cantor string $\mathcal L=(l_j)_{j\ge1}$ corresponding to the set $A$, each $l_j=3^{-j}$ has multiplicity $2^{j-1}$, for each $j\ge1$ (where $(l_j)_{j\ge1}$ denotes the sequence of {\em distinct} lengths of $\mathcal L$). Therefore, as is well known,
\begin{equation}\label{4.351/2}
\zeta_{\mathcal L}(s)=\sum_{j=1}^\ty2^{j-1}3^{-js}=\frac{3^{-s}}{1-2\cdot 3^{-s}}.
\end{equation} 
 Hence, $\zeta_{\mathcal L}$ has a meromorphic extension to all of $\Ce$, given by the last expression in (\ref{4.351/2}). It now follows from (\ref{cantor_string}) and (\ref{equ_tilde}) that both the distance and the tube zeta functions of $A$ possess a meromorphic extension to the entire complex plane, with one additional simple pole at~$s=0$.
\end{example}

\begin{remark}\label{s=0}
Throughout this section (i.e., \S\ref{merom_ext}),
we use the original notion of complex dimension of a bounded fractal string, namely,
a pole of the meromorphic continuation of the geometric zeta function (see [Lap-vFr1--2]).
If, instead, we used the notion of complex dimension from the general theory developed in this paper
(and in [LapRa\v Zu1--8]), that is, a pole of the associated fractal zeta function (distance or tube zeta function), then we should add $s=0$ as a simple complex dimension throughout \S\ref{merom_ext} and \S\ref{hyperfractals} when discussing the set of complex dimensions of a fractal string. The principal complex dimensions, however, which are the main focus in the present discussion, would remain unchanged.
\end{remark}

\begin{example}\label{4.24}
The asymptotics of the tube function $t\mapsto|A_t|$, as given in (\ref{G}), arise naturally in the study of self-similar lattice strings; see \cite[\S8.4.4, esp.,\ Eq.\ (8.44)]{lapidusfrank12}.
In light of the results of \cite[Ch.\ 3 and \S8.4]{lapidusfrank12}, an analogous comment can be made about self-similar sprays or tilings in $\eR^N$, with $N\ge1$ arbitrary; see [LapPeWi].
Furthermore, it is expected to also hold for lattice self-similar sets in $\eR^N$ satisfying the open set condition.
\end{example}

Assuming that $D<N$ and using the functional equation \eqref{equ_tilde} connecting $\zeta_A$ and $\tilde\zeta_A$, Theorems \ref{measurable} and \ref{nonmeasurable} can be easily restated in terms of the distance zeta functions, thereby providing a significant extension of Theorem \ref{pole1}.

\begin{theorem}[Distance zeta functions of bounded sets: Minkowski measurable and nonmeasurable cases]\label{meas_nonmeas}%
Let $A$ be a bounded subset of $\eR^N$, with $N\geq 1$.
In the Minkowski measurable case, we assume that hypothesis~\eqref{A_t} of Theorem~\ref{measurable} holds, 
while in the Minkowski nonmeasurable case, we assume that hypothesis~\eqref{G} of Theorem~\ref{nonmeasurable} holds. Furthermore, let $D\geq 0$ be the real number occurring in~\eqref{A_t} or in~\eqref{G}, respectively.
Then, if $D<N$, the conclusions of Theorem~\ref{measurable} 
$($resp., of Theorem~\ref{nonmeasurable}$)$, concerning the tube zeta function $(\tilde{\zeta}_A=\tilde{\zeta}_A(\,\cdot\,,\delta)$ for any fixed $\delta>0)$, also hold for the 
distance zeta function $\zeta_A$,
except for the values of the residues. In the case of the counterpart of Theorem~\ref{measurable}, these values are given by
\begin{equation}\label{Mres}
\res(\zeta_A,D)=(N-D)\M^D(A)
\end{equation}
and, more generally, in the case of the counterpart 
of Theorem~\ref{nonmeasurable}, by
\begin{equation}\label{NMres}
\res(\zeta_A,s_k)=\frac{N-s_k}T\hat G_0\Big(\frac kT\Big),
\end{equation}
for each pole $s_k\in\mathcal {P}(\tilde\zeta_A)$, with $k\in\Ze$; see Eq.\ $(\ref{Dpoles})$.
Furthermore,
\begin{equation}\label{NMres_asymptotics}
|\res(\zeta_A,s_k)|= o(|k|)\q \mathrm{as}\q k\to\pm\ty. 
\end{equation}

Moreover,
\begin{equation}\label{avarage1}
\res(\zeta_A,D)=(N-D)\frac1T\int_0^T G(\tau)\,\D\tau
\end{equation}
and
\begin{equation}\label{NMres_ineq}
(N-D)\M_*^D(A)<\res(\zeta_A,D)<(N-D)\M^{*D}(A).
\end{equation}
\end{theorem}

\subsection{Generating principal complex dimensions of arbitrary order}\label{gen}

The next example shows how one can effectively construct bounded fractal strings having principal complex dimensions of any given order (i.e., multiplicity), and even an infinite number of essential singularities on the critical line.
In the latter case, we deduce that the corresponding geometric zeta function cannot be meromorphically extended to any domain of $\Ce$ containing the critical line.

\begin{example}[Cantor strings of higher order]\label{high_order}
We will provide an example of a bounded fractal string $\mathcal L$ in $\eR$ such that its geometric zeta function has an infinite sequence of principal complex dimensions of higher order and in arithmetic progression.
The construction is based on an `iterated' Cantor string.
Let $C$ be the standard middle-third Cantor set contained in $[0,1]$.
Then the open set $\O:=(0,1)\setminus C$ is a geometric realization of the Cantor string ${\mathcal L}={\mathcal L}_{CS}$; that is, the open connected components of $\O$ are intervals whose lengths (counting
multiplicities) are those listed in $\mathcal L$, namely, $1/3$, $1/9,1/9$, \dots, $1/3^n$ (with multiplicity $2^{n-1}$), for any integer $n\ge1$. The Cantor string ${\mathcal L}$ can be interpreted as a self-similar fractal spray with generator the interval $(0,1)$ and with two scaling ratios $r_1=r_2=1/3$.\footnote{See Definition \ref{ss_spray} below, where ${\mathcal L}_0$ is taken to be the trivial fractal string associated with the single interval $(0,1)$, for the definition of a self-similar spray.}

Now, let $\O=(0,1)\setminus C$ be our generator for the new self-similar fractal spray in $\eR$ with scaling ratios $r_1=r_2=1/3$.
We denote the corresponding fractal string by ${\mathcal L}_2$ (see Figure~\ref{second_cantor}), and its geometric realization by $\O_2$.
Then, by the scaling property of the geometric zeta function (see Eq.\  \eqref{zetalA} in Theorem \ref{an}$(d)$), we have
$$
\zeta_{{\mathcal L}_2}(s)=\zeta_{\mathcal L}(s)+2\,\zeta_{3^{-1}{\mathcal L}_2}(s)=\zeta_{\mathcal L}(s)+2\cdot 3^{-s}\,\zeta_{{\mathcal L}_2}(s).
$$
According to Eq.\  \eqref{4.351/2}, we conclude that
\begin{equation}
\zeta_{{\mathcal L}_2}(s)=\frac{3^s}{3^s-2}\zeta_{\mathcal L}(s)=\frac{3^s}{(3^s-2)^2}
\end{equation}
for $\re s>\log_32$.
Hence, $\zeta_{{\mathcal L}_2}$ can be meromorphically extended to all of $\Ce$ and
\begin{equation}\label{cantorn}
\po(\zeta_{{\mathcal L}_2})=\log_32+\frac{2\pi}{\log 3}\I\Ze.
\end{equation}
Furthermore, the poles $\omega_k:=\log_32+\frac{2\pi}{\log 3}k\I$ of $\zeta_{{\mathcal L}_2}$ (i.e., the complex dimensions of ${\mathcal L}_2$
) for $k\in\Ze$ are all of second order (i.e., of multiplicity two).
We deduce, in particular, that $\ov\dim_B{\mathcal L}_2=\log_32$.

%

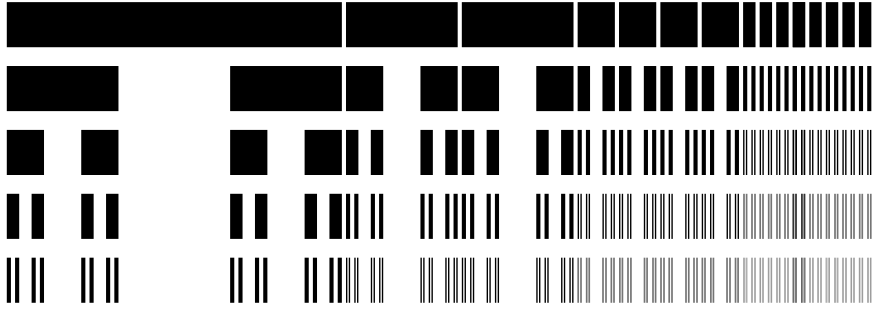
\begin{figure}[ht]
\begin{center}
\hspace{-2.5mm}
\begin{tikzpicture}[decoration=Cantor set,line width=6mm,xscale=0.055,yscale=1.7]
  \draw (0,0) -- (81,0);
  \draw decorate{ (0,-.5) -- (81,-.5) };
  \draw decorate{ decorate{ (0,-1) -- (81,-1) }};
  \draw decorate{ decorate{ decorate{ (0,-1.5) -- (81,-1.5) }}};
  \draw decorate{ decorate{ decorate{ decorate{ (0,-2) -- (81,-2) }}}};
    \draw (81+1,0) -- (108+1,0);
    \draw decorate{ (81+1,-.5) -- (81+27+1,-.5) };
  \draw decorate{ decorate{ (81+1,-1) -- (81+27+1,-1) }};
  \draw decorate{ decorate{ decorate{ (81+1,-1.5) -- (81+27+1,-1.5) }}};
  \draw decorate{ decorate{ decorate{ decorate{ (81+1,-2) -- (81+27+1,-2) }}}};
    \draw (81+27+2,0) -- (81+27+27+2,0);
    \draw decorate{ (81+27+2,-.5) -- (81+27+27+2,-.5) };
  \draw decorate{ decorate{ (81+27+2,-1) -- (81+27+27+2,-1) }};
  \draw decorate{ decorate{ decorate{ (81+27+2,-1.5) -- (81+27+27+2,-1.5) }}};
  \draw decorate{ decorate{ decorate{ decorate{ (81+27+2,-2) -- (81+27+27+2,-2) }}}};
    \draw (81+27+27+3,0) -- (81+27+27+9+3,0);
    \draw decorate{ (81+27+27+3,-.5) -- (81+27+27+9+3,-.5) };
  \draw decorate{ decorate{ (81+27+27+3,-1) -- (81+27+27+9+3,-1) }};
  \draw decorate{ decorate{ decorate{ (81+27+27+3,-1.5) -- (81+27+27+9+3,-1.5) }}};
  \draw decorate{ decorate{ decorate{ decorate{ (81+27+27+3,-2) -- (81+27+27+9+3,-2) }}}};
    \draw (81+27+27+9+4,0) -- (81+27+27+9+9+4,0);
    \draw decorate{ (81+27+27+9+4,-.5) -- (81+27+27+9+9+4,-.5) };
  \draw decorate{ decorate{ (81+27+27+9+4,-1) -- (81+27+27+9+9+4,-1) }};
  \draw decorate{ decorate{ decorate{ (81+27+27+9+4,-1.5) -- (81+27+27+9+9+4,-1.5) }}};
  \draw decorate{ decorate{ decorate{ decorate{ (81+27+27+9+4,-2) -- (81+27+27+9+9+4,-2) }}}};
    \draw (81+27+27+9+9+5,0) -- (81+27+27+9+9+9+5,0);
    \draw decorate{ (81+27+27+9+9+5,-.5) -- (81+27+27+9+9+9+5,-.5) };
  \draw decorate{ decorate{ (81+27+27+9+9+5,-1) -- (81+27+27+9+9+9+5,-1) }};
  \draw decorate{ decorate{ decorate{ (81+27+27+9+9+5,-1.5) -- (81+27+27+9+9+9+5,-1.5) }}};
  \draw decorate{ decorate{ decorate{ decorate{ (81+27+27+9+9+5,-2) -- (81+27+27+9+9+9+5,-2) }}}};
    \draw (81+27+27+9+9+9+6,0) -- (81+27+27+9+9+9+9+6,0);
    \draw decorate{ (81+27+27+9+9+9+6,-.5) -- (81+27+27+9+9+9+9+6,-.5) };
  \draw decorate{ decorate{ (81+27+27+9+9+9+6,-1) -- (81+27+27+9+9+9+9+6,-1) }};
  \draw decorate{ decorate{ decorate{ (81+27+27+9+9+9+6,-1.5) -- (81+27+27+9+9+9+9+6,-1.5) }}};
  \draw decorate{ decorate{ decorate{ decorate{ (81+27+27+9+9+9+6,-2) -- (81+27+27+9+9+9+9+6,-2) }}}};
    \draw (81+27+27+9+9+9+9+7,0) -- (81+27+27+9+9+9+9+3+7,0);
    \draw decorate{ (81+27+27+9+9+9+9+7,-.5) -- (81+27+27+9+9+9+9+3+7,-.5) };
  \draw decorate{ decorate{ (81+27+27+9+9+9+9+7,-1) -- (81+27+27+9+9+9+9+3+7,-1) }};
  \draw decorate{ decorate{ decorate{ (81+27+27+9+9+9+9+7,-1.5) -- (81+27+27+9+9+9+9+3+7,-1.5) }}};
  \draw decorate{ decorate{ decorate{ decorate{ (81+27+27+9+9+9+9+7,-2) -- (81+27+27+9+9+9+9+3+7,-2) }}}};
    \draw (81+27+27+9+9+9+9+3+8,0) -- (81+27+27+9+9+9+9+3+3+8,0);
    \draw decorate{ (81+27+27+9+9+9+9+3+8,-.5) -- (81+27+27+9+9+9+9+3+3+8,-.5) };
  \draw decorate{ decorate{ (81+27+27+9+9+9+9+3+8,-1) -- (81+27+27+9+9+9+9+3+3+8,-1) }};
  \draw decorate{ decorate{ decorate{ (81+27+27+9+9+9+9+3+8,-1.5) -- (81+27+27+9+9+9+9+3+3+8,-1.5) }}};
  \draw decorate{ decorate{ decorate{ decorate{ (81+27+27+9+9+9+9+3+8,-2) -- (81+27+27+9+9+9+9+3+3+8,-2) }}}};
    \draw (81+27+27+9+9+9+9+3+3+9,0) -- (81+27+27+9+9+9+9+3+3+3+9,0);
    \draw decorate{ (81+27+27+9+9+9+9+3+3+9,-.5) -- (81+27+27+9+9+9+9+3+3+3+9,-.5) };
  \draw decorate{ decorate{ (81+27+27+9+9+9+9+3+3+9,-1) -- (81+27+27+9+9+9+9+3+3+3+9,-1) }};
  \draw decorate{ decorate{ decorate{ (81+27+27+9+9+9+9+3+3+9,-1.5) -- (81+27+27+9+9+9+9+3+3+3+9,-1.5) }}};
  \draw decorate{ decorate{ decorate{ decorate{ (81+27+27+9+9+9+9+3+3+9,-2) -- (81+27+27+9+9+9+9+3+3+3+9,-2) }}}};
    \draw (81+27+27+9+9+9+9+3+3+3+10,0) -- (81+27+27+9+9+9+9+3+3+3+3+10,0);
    \draw decorate{ (81+27+27+9+9+9+9+3+3+3+10,-.5) -- (81+27+27+9+9+9+9+3+3+3+3+10,-.5) };
  \draw decorate{ decorate{ (81+27+27+9+9+9+9+3+3+3+10,-1) -- (81+27+27+9+9+9+9+3+3+3+3+10,-1) }};
  \draw decorate{ decorate{ decorate{ (81+27+27+9+9+9+9+3+3+3+10,-1.5) -- (81+27+27+9+9+9+9+3+3+3+3+10,-1.5) }}};
  \draw decorate{ decorate{ decorate{ decorate{ (81+27+27+9+9+9+9+3+3+3+10,-2) -- (81+27+27+9+9+9+9+3+3+3+3+10,-2) }}}};
    \draw (81+27+27+9+9+9+9+3+3+3+10,0) -- (81+27+27+9+9+9+9+3+3+3+3+10,0);
    \draw decorate{ (81+27+27+9+9+9+9+3+3+3+10,-.5) -- (81+27+27+9+9+9+9+3+3+3+3+10,-.5) };
  \draw decorate{ decorate{ (81+27+27+9+9+9+9+3+3+3+10,-1) -- (81+27+27+9+9+9+9+3+3+3+3+10,-1) }};
  \draw decorate{ decorate{ decorate{ (81+27+27+9+9+9+9+3+3+3+10,-1.5) -- (81+27+27+9+9+9+9+3+3+3+3+10,-1.5) }}};
  \draw decorate{ decorate{ decorate{ decorate{ (81+27+27+9+9+9+9+3+3+3+10,-2) -- (81+27+27+9+9+9+9+3+3+3+3+10,-2) }}}};
    \draw (81+27+27+9+9+9+9+3+3+3+3+11,0) -- (81+27+27+9+9+9+9+3+3+3+3+3+11,0);
    \draw decorate{ (81+27+27+9+9+9+9+3+3+3+3+11,-.5) -- (81+27+27+9+9+9+9+3+3+3+3+3+11,-.5) };
  \draw decorate{ decorate{ (81+27+27+9+9+9+9+3+3+3+3+11,-1) -- (81+27+27+9+9+9+9+3+3+3+3+3+11,-1) }};
  \draw decorate{ decorate{ decorate{ (81+27+27+9+9+9+9+3+3+3+3+11,-1.5) -- (81+27+27+9+9+9+9+3+3+3+3+3+11,-1.5) }}};
  \draw decorate{ decorate{ decorate{ decorate{ (81+27+27+9+9+9+9+3+3+3+3+11,-2) -- (81+27+27+9+9+9+9+3+3+3+3+3+11,-2) }}}};
    \draw (81+27+27+9+9+9+9+3+3+3+3+3+12,0) -- (81+27+27+9+9+9+9+3+3+3+3+3+3+12,0);
    \draw decorate{ (81+27+27+9+9+9+9+3+3+3+3+3+12,-.5) -- (81+27+27+9+9+9+9+3+3+3+3+3+3+12,-.5) };
  \draw decorate{ decorate{ (81+27+27+9+9+9+9+3+3+3+3+3+12,-1) -- (81+27+27+9+9+9+9+3+3+3+3+3+3+12,-1) }};
  \draw decorate{ decorate{ decorate{ (81+27+27+9+9+9+9+3+3+3+3+3+12,-1.5) -- (81+27+27+9+9+9+9+3+3+3+3+3+3+12,-1.5) }}};
  \draw decorate{ decorate{ decorate{ decorate{ (81+27+27+9+9+9+9+3+3+3+3+3+12,-2) -- (81+27+27+9+9+9+9+3+3+3+3+3+3+12,-2) }}}};
    \draw (81+27+27+9+9+9+9+3+3+3+3+3+3+13,0) -- (81+27+27+9+9+9+9+3+3+3+3+3+3+3+13,0);
    \draw decorate{ (81+27+27+9+9+9+9+3+3+3+3+3+3+13,-.5) -- (81+27+27+9+9+9+9+3+3+3+3+3+3+3+13,-.5) };
  \draw decorate{ decorate{ (81+27+27+9+9+9+9+3+3+3+3+3+3+13,-1) -- (81+27+27+9+9+9+9+3+3+3+3+3+3+3+13,-1) }};
  \draw decorate{ decorate{ decorate{ (81+27+27+9+9+9+9+3+3+3+3+3+3+13,-1.5) -- (81+27+27+9+9+9+9+3+3+3+3+3+3+3+13,-1.5) }}};
  \draw decorate{ decorate{ decorate{ decorate{ (81+27+27+9+9+9+9+3+3+3+3+3+3+13,-2) -- (81+27+27+9+9+9+9+3+3+3+3+3+3+3+13,-2) }}}};
    \draw (81+27+27+9+9+9+9+3+3+3+3+3+3+3+14,0) -- (81+27+27+9+9+9+9+3+3+3+3+3+3+3+3+14,0);
    \draw decorate{ (81+27+27+9+9+9+9+3+3+3+3+3+3+3+14,-.5) -- (81+27+27+9+9+9+9+3+3+3+3+3+3+3+3+14,-.5) };
  \draw decorate{ decorate{ (81+27+27+9+9+9+9+3+3+3+3+3+3+3+14,-1) -- (81+27+27+9+9+9+9+3+3+3+3+3+3+3+3+14,-1) }};
  \draw decorate{ decorate{ decorate{ (81+27+27+9+9+9+9+3+3+3+3+3+3+3+14,-1.5) -- (81+27+27+9+9+9+9+3+3+3+3+3+3+3+3+14,-1.5) }}};
  \draw decorate{ decorate{ decorate{ decorate{ (81+27+27+9+9+9+9+3+3+3+3+3+3+3+14,-2) -- (81+27+27+9+9+9+9+3+3+3+3+3+3+3+3+14,-2) }}}};
\end{tikzpicture}
\end{center}
\caption{The {\em second order Cantor set} from Example~\ref{high_order}. Only the first four iterations are shown here. More precisely, from left to right, we have the middle-third Cantor set $C$ in $[0,1]$, then two copies of $C$ scaled by $1/3$, and then four copies of $C$ scaled by $1/9$, and so on, ad infinitum.}\label{second_cantor}
\end{figure}

We can now repeat the above process inductively; that is, for any integer $n\ge2$, we define the $n$-{\em th order Cantor string} ${\mathcal L}_n$ via a self-similar fractal spray (or ``extended self-similar fractal string'', in the terminology of Definition \ref{ss_spray} below) generated by ${\mathcal L}_{n-1}$ and with the same scaling ratios $r_1=r_2=1/3$. (Note that by construction, we have ${\mathcal L}_1:=\mathcal L$, the Cantor string.) We denote the resulting geometric realization of the fractal string ${\mathcal L}_n$ by $\O_n$.
Similarly as before, we have that  $\zeta_{{\mathcal L}_n}$ has a meromophic extension to all of $\Ce$ satisfying
$$
\zeta_{{\mathcal L}_n}(s)=\frac{3^{(n-1)s}}{(3^s-2)^n}.
$$
The poles $\omega_k$ of $\zeta_{{\mathcal L}_n}$ (that is, the complex dimensions of 
${\mathcal L}_n$
) are of order $n$ and $D:=\ov\dim_B{\mathcal L}_n=\log_32$.
More specifically, viewed as a set, $\po(\zeta_{{\mathcal L}_n})$ is given by Eq.\  \eqref{cantorn}, but each complex dimension $\o_k=\log_32+\frac{2\pi}{\log}k\I$ (for $k\in\Ze$) is now of multiplicity~$n$. 



Finally, we can now use the $n$-th order Cantor strings $\mathcal{L}_n$ in order to construct an infinite order Cantor string ${\mathcal L}_\ty$ which will have infinitely many essential singularities on the critical line.
First of all, we let ${\mathcal L}_1:=\mathcal{L}$ be the Cantor string, as before, then scale down every fractal string ${\mathcal L}_n$ by the factor $3^{-n}/n!$, for each $n\ge1$, and we define ${\mathcal L}_\ty$ as the disjoint union of the resulting fractal strings (viewed as multisets; that is, taking multiplicities into account; see the end of \S\ref{notation}):
$$
{\mathcal L}_\ty:=\bigsqcup_{n=1}^{\ty} \frac{3^{-n}}{n!}{\mathcal L}_n.
$$
We then have
\begin{equation}\label{ess_zeta}
\begin{aligned}
\zeta_{{\mathcal L}_\ty}(s)&=\sum_{n=1}^{\ty}\zeta_{3^{-n}(n!)^{-1}{\mathcal L}_n}(s)=\sum_{n=1}^{\ty}\frac{3^{-ns}}{(n!)^s}\zeta_{{\mathcal L}_n}(s)\\
&=\frac{1}{3^s}\sum_{n=1}^{\ty}\frac{1}{(n!)^s(3^s-2)^n}.
\end{aligned}
\end{equation}
By the Weierstrass $M$-test, $\zeta_{{\mathcal L}_\ty}$ is holomorphic on the connected open set $\{\re s>0\}\setminus\big(\log_32+\frac{2\pi}{\log 3}\I\Ze\big)$.
Furthermore, it has an essential singularity at each point of the discrete set $\log_32+\frac{2\pi}{\log 3}\I\Ze$. In particular, we have that 
\begin{equation}
D_{\rm mer}(\zeta_{{\mathcal L}_\ty})=D_{\rm hol}(\zeta_{{\mathcal L}_\ty})=D(\zeta_{{\mathcal L}_\ty})=\log_32.
\end{equation}
The critical line $\{\re s=\log_32\}$ is clearly not a natural boundary for $\zeta_{{\mathcal L}_\ty}$ since $\zeta_{{\mathcal L}_\ty}$ (given by $\zeta_{{\mathcal L}_\ty}(s):=3^{-s}\sum_{n=1}^{\ty}(n!)^{-s}(3^s-2)^{-n}$; see \eqref{ess_zeta}) can be holomorphically continued to the connected open set $\{\re s>0\}\setminus\dim_{PC}{\mathcal L}_\ty$.

\end{example}

Each $n$-th order Cantor string ${\mathcal L}_n$ can be identified with the set $A_{{\mathcal L}_n}$ described in Eq.\  \eqref{AL} of Example~\ref{L}, which we call the $n$-{\em th order Cantor set}.

The above construction can be generalized verbatim to any (nontrivial) bounded fractal string $\mathcal{L}$, instead of the initial Cantor string ${\mathcal L}_1={\mathcal L}_{CS}$.
This suggests that the definition of complex dimensions should be extended to also include essential singularities of geometric zeta functions (or, more generally, of fractal zeta functions), in the spirit of \cite[\S12.1 and \S13.4.3]{lapidusfrank12} and \cite{fzf}). 

Let us now recall the definition of a self-similar fractal string (see 
\cite[\S2.1]{lapidusfrank12}).
In fact, we introduce a somewhat more general notion.

\begin{defn}\label{ss_spray} Let ${\mathcal L}_0$ be a bounded fractal string and $\{r_1,\ldots,r_J\}$ be a multiset of positive numbers (``ratio list'' or scaling ratios) such that 
$
\sum_{j=1}^Jr_j<1
$.
An {\em extended self-similar fractal string} $\ov{\mathcal L}=\ov{\mathcal L}({\mathcal L}_0;r_1,\dots,r_J)$, generated by ${\mathcal L}_0$ and $\{r_1,r_2,\ldots,r_J\}$, is the bounded fractal string defined by
\begin{equation}\label{ovL}
\ov{\mathcal L}:=\bigsqcup_{\a\in(\eN_0)^J}(r_1^{\a_1}\dots r_J^{\a_J})\mathcal L_0,
\end{equation}
where $\a:=(\a_1,\dots,\a_J)$ and the notation $\sqcup$ is described towards the end of \S\ref{notation}. Therefore, $\ov{\mathcal L}$ can be written as the following tensor product of fractal strings:
\begin{equation}
\ov{\mathcal L}={\mathcal L}_0\otimes{\mathcal L}(r_1,\dots,r_J),
\end{equation}
where the tensor product is defined at the end of \S\ref{notation} and the fractal string $\mathcal L(r_1,\dots,r_J)$ is defined by
\begin{equation}
\mathcal L(r_1,\dots,r_J):=\{r_1^{\a_1}\dots r_J^{\a_J}:\a\in(\eN_0)^J\},
\end{equation}
viewed as a multiset (i.e., taking multiplicities into account).
\end{defn}

The proof of the following simple lemma is omitted.

\begin{lemma}\label{tensorl} Let $\mathcal L_1$ and $\mathcal L_2$ be two bounded fractal strings. Then, the geometric zeta function of their tensor product is given by
\begin{equation}\label{tensorl=}
\zeta_{{\mathcal L_1}\otimes{\mathcal L_2}}(s)=\zeta_{\mathcal L_1}(s)\cdot\zeta_{\mathcal L_2}(s)
\end{equation}
for $\re s>\max\{\ov\dim_B {\mathcal L}_1,\ov\dim_B {\mathcal L}_2\}$. Furthermore,
\begin{equation}\label{tensorld}
\begin{aligned}
\ov\dim_B({\mathcal L}_1\otimes{\mathcal L}_2)&=\max\{\ov\dim_B {\mathcal L}_1,\ov\dim_B {\mathcal L}_2\};\ \ i.e.,\\
D({\mathcal L}_1\otimes{\mathcal L}_2)&=\max\{D({\mathcal L}_1),D({\mathcal L}_2)\}.
\end{aligned}
\end{equation}
\end{lemma}

\medskip

The following theorem extends \cite[Thm.\ 2.3]{lapidusfrank12} to the present more general context of extended self-similar strings.
\medskip

\begin{theorem}\label{tmab}
Let the assumptions of Definition \ref{ss_spray} be satisfied and let $D\in(0,1)$ be the $($necessarily unique$)$ {\rm real} solution of the Moran equation $\sum_{j=1}^Jr_j^s=1$. 


Then, the extended self-similar fractal string $\ov{\mathcal L}:={\mathcal L}_0\otimes{\mathcal L}(r_1,\dots,r_J)$ is bounded and has total length given by
$
|\ov{\mathcal L}|_1={|{\mathcal L}_0|_1}/{(1-\sum_{j=1}^Jr_j)}.
$
Further, its geometric zeta function has for abscissa of meromorphic continuation $D_{\rm mer}(\zeta_{\ov{\mathcal L}})=D_{\rm mer}(\zeta_{{\mathcal L}_0})$ and its meromorphic extension to all of $\Ce$ is given by
\begin{equation}\label{zovL}
\zeta_{\ov{\mathcal L}}(s)=\frac{\zeta_{{\mathcal L}_0}(s)}{1-\sum_{j=1}^J r_j^s}.
\end{equation}
Moreover, its abscissa of convergence is given by
$
D(\zeta_{\ov{\mathcal L}})=\max\{D(\zeta_{{\mathcal L}_0}),D\}.
$

Finally, for a given domain $U$ containing the critical line $\{\re s=D(\zeta_{\mathcal{L}_0})\}$ of $\zeta_{{\mathcal L}_0}$, the visible complex dimensions in $U$ of  $\ov{\mathcal L}$ satisfy
\begin{equation}\label{ss_spray_po}
\po(\zeta_{\ov{\mathcal L}},U)\subseteq\mathfrak{D}\cup\po(\zeta_{{\mathcal L}_0},U),
\end{equation}
where $\mathfrak{D}$ is the set of complex solutions in $U$ of the Moran equation $\sum_{j=1}^{J}r_j^s=1$.
Furthermore, if there are no zero-pole cancellations in \eqref{zovL},\footnote{This happens, for example, if ${\mathcal L}_0$ is the trivial fractal string ${\mathcal L}_0:=\{1\}$, so that $\zeta_{{\mathcal L}_0}(s)\equiv 1$.} then we have an equality in~ \eqref{ss_spray_po}.
\end{theorem}

\begin{proof}
The proof is based on a scaling argument, as follows.
Clearly, we have
\begin{equation}
\ov{\mathcal L}={\mathcal L}_0\sqcup\bigsqcup_{j=1}^J r_j\ov{\mathcal L}.
\end{equation}
Hence, the geometric zeta function of $\ov{\mathcal L}$ satisfies the following functional equation:
\begin{equation}
\zeta_{\ov{\mathcal L}}(s)=\zeta_{{\mathcal L}_0}(s)+\sum_{j=1}^J\zeta_{r_j\ov{\mathcal L}}(s).
\end{equation}
Further, by using the scaling property of the geometric zeta function (see Eq.\  \eqref{zetalA}  of Theorem \ref{an}$(d)$ above, stated there for the distance zeta function but also valid for the geometric zeta function, as is well known), the above equation becomes
\begin{equation}\label{4.455}
\zeta_{\ov{\mathcal L}}(s)=\zeta_{\ov{\mathcal L}_0}(s)+\zeta_{\ov{\mathcal L}}(s)\sum_{j=1}^Jr_j^s.
\end{equation}
Since the series defining $\zeta_{\ov{\mathcal L}}(s)$ 
is convergent for $\re s>1$, Eq.\ \eqref{4.455} yields Eq.\  \eqref{zovL} directly for $\re s>D_{\rm mer}(\zeta_{{\mathcal L}_0})$. Indeed, upon meromorphic continuation, each of the meromophic functions involved in the above scaling argument can be interpreted as the meromorphic continuation of the corresponding zeta function.
%
\end{proof}





The next theorem gives a general construction of complex dimensions of higher order generated by means of extended self-similar strings.

\begin{theorem}\label{higher_order_dim}
Let $\ov{\mathcal L}:={\mathcal L}_0\otimes{\mathcal L}(r_1,\dots,r_J)$ be an extended self-similar fractal string in $\eR$ generated by a bounded fractal string ${\mathcal L}_0$ and the set of scaling ratios $\{r_1,r_2,\ldots,r_J\}$ with $0<r_j<1$, for $j=1,\dots,J$, such that $\sum_{j=1}^J r_j<1$.
Furthermore, assume that $\zeta_{{\mathcal L}_0}$ is meromorphic on $\Ce$ and that there are no zero-pole cancellations in \eqref{zovL}.
Let $\mathfrak{D}$ be the set of complex solutions of the Moran equation $\sum_{j=1}^{J}r_j^s=1$ and let $m$ be an arbitrary positive integer.
Then, one can explicitly construct an extended self-similar fractal string $\ov{\mathcal L}_m$ which has exactly the same complex dimensions as $\ov{\mathcal L}$ but with the orders $($i.e., the multiplicities$)$ of the complex dimensions located in $\mathfrak{D}$ multiplied by $m$.

Moreover, if we let $\mathfrak{D}^+:=\mathfrak{D}\cap\{\re s>0\}$, then one can explicitly construct an extended self-similar fractal string $\ov{\mathcal L}_{\ty}$ such that all of its complex dimensions contained in $\mathfrak{D}^+$ are of infinite order; that is, they are essential singularities of its geometric zeta function $\zeta_{\ov{\mathcal L}_{\ty}}$. In particular, we have that $D_{\rm mer}(\zeta_{\ov{\mathcal L}_{\ty}})=D(\zeta_{\ov{\mathcal L}_{\ty}})$.
\end{theorem}

\begin{proof}
Let ${\mathcal L}_0$ be the generator and let $\{r_1,r_2,\ldots,r_J\}$ be the associated scaling ratios, as
in the statement of the theorem.
Furthermore, we define $\ov{\mathcal L}:=\sqcup_{k=1}^\ty{\mathcal L}_k$,  as in Theorem \ref{tmab}, and we now let this be our new generator; that is, we define a new extended self-similar fractal string $\ov{\mathcal L}_2$ as the disjoint union of scaled copies of $\ov{\mathcal L}$ by scaling factors built by all possible words of multiples of the ratios $r_j$.
This construction implies that
$
\ov{\mathcal L}_2=\ov{\mathcal L}\sqcup \bigsqcup_{j=1}^J r_j \ov{\mathcal L}_2;
$
similarly as before, by the scaling property of the geometric zeta function and, in light of Theorem \ref{tmab}, we then have
\begin{equation}\label{zovL1}
\zeta_{\ov{\mathcal L}_2}(s)=\frac{\zeta_{\ov{\mathcal L}}(s)}{1-\sum_{j=1}^{J}r_j^s}=\frac{\zeta_{{\mathcal L}_0}(s)}{\big(1-\sum_{j=1}^{J}r_j^s\big)^2}.
\end{equation}
As is apparent in Eq.\  \eqref{zovL1}, the fractal string $\ov{\mathcal L}_2$ has exactly the same complex dimensions as $\ov{\mathcal L}$, except for the fact that the orders of the ones contained in $\mathfrak D$ are multiplied by $2$.

We can next proceed inductively by using $\ov{\mathcal L}_2$ as our new base fractal string and, for each $n\in\eN$, we thus obtain a fractal string $\ov{\mathcal L}_n$ such that
\begin{equation}
\zeta_{\ov{\mathcal L}_n}(s)=\frac{\zeta_{{\mathcal L}_0}(s)}{\big(1-\sum_{j=1}^{J}r_j^s\big)^n};
\end{equation}
hence, $\ov{\mathcal L}_n$ has exactly the same complex dimensions as $\ov{\mathcal L}$, except for the fact that the ones contained in $\mathfrak D$ have their orders multiplied by $n$.

In order to generate essential singularities, we take a disjoint union of the fractal strings $\ov{\mathcal L}_n$ scaled by $(n!)^{-1}$.
More specifically, we define $\ov{\mathcal L}_\ty$ as
\begin{equation}
\ov{\mathcal L}_\ty:=\bigsqcup_{n=1}^{\ty}(n!)^{-1}\ov{\mathcal L}_n.
\end{equation}
The construction of $\ov{\mathcal L}_\ty$ and the scaling property of the geometric zeta function then imply that
\begin{equation}
\zeta_{\ov{\mathcal L}_\ty}(s)=\zeta_{{\mathcal L}_0}(s)\sum_{n=1}^{\ty}\frac{1}{(n!)^s\big(1-\sum_{j=1}^{J}r_j^s\big)^n}.
\end{equation}
By the Weierstrass $M$-test, the above sum is seen to define a holomorphic function on $\{\re s>0\}\setminus\mathfrak{D}^+$  and $\mathfrak{D}^+$ is the set of essential singularities of the function defined by this sum, as desired.
\end{proof}


\section{A construction of a class of maximally hyperfractal sets}\label{hyperfractals}

A subset $A$ of $\eR^N$ is said to be {\em maximally hyperfractal} if the critical line $\{\re s=D(\zeta_A)\}$ of the corresponding distance zeta function $\zeta_A$ consists solely of nonremovable singularities of $\zeta_A$. In particular, this implies that $D_{\rm mer}(\zeta_A)=D_{\rm hol}(\zeta_A)=D(\zeta_A)=\ov\dim_BA$, since $\zeta_A$ cannot be meromorphically extended to an open right half-plane $\{\re s>\a\}$ with $\a<D(\zeta_A)$ and for any bounded set $A$, we have $D_{\rm mer}(\zeta_A)\le D_{\rm hol}(\zeta_A)\le D(\zeta_A)=\ov\dim_BA$; see Corollary \ref{zetac}$(a)$. 

We shall need a class of generalized Cantor sets depending on two parameters, which we now introduce.

\begin{defn}\label{Cma}
The generalized Cantor sets $C^{(m,a)}$ are determined by an integer $m\ge2$ and a positive real number $a$ such that $ma<1$.
In the first step of the analog of Cantor's construction, we start with $m$ equidistant, closed intervals in $[0,1]$ of length $a$, with $m-1$ holes, each of length $(1-ma)/(m-1)$. In the second step, we continue by scaling by the factor $a$ each of the $m$ intervals of length $a$; and so on, ad infinitum.
The  $($two-parameter$)$ {\em generalized Cantor set} $C^{(m,a)}$ is then defined as the intersection of the decreasing sequence of compact sets constructed in this way.
\end{defn}

In the following proposition, we collect some of the basic properties of generalized Cantor sets $C^{(m,a)}$.
Apart from the proof of \eqref{zetaCma}, which is easily obtained via a direct computation, the proof of the proposition
is similar to that for the standard ternary Cantor set or string (see \cite[Eq.\ (1.11)]{lapidusfrank12}), and therefore, we omit it.

\begin{prop}\label{Cmap}
 If $A:=C^{(m,a)}\subseteq\eR$ is the generalized Cantor set introduced in Definition~\ref{Cma}, then
$
D:=\dim_B C^{(m,a)}=D(\zeta_A)=\log_{1/a}m.
$
Also, $A$ is Minkowski nondegenerate but is not Minkowski measurable. 
If we assume that $\delta\ge\frac{1-ma}{2(m-1)}$, then
\begin{equation}\label{zetaCma}
\zeta_A(s):=\int_{-\delta}^{1+\delta}d(x,A)^{s-2}\D x=\left(\frac{1-ma}{2(m-1)}\right)^{s-1}\frac{1-ma}{s(1-ma^s)}+\frac{2\delta^s}s.
\end{equation}
As a result, $\zeta_A(s)$ admits a meromorphic continuation to all of $\Ce$, given by the last expression in $(\ref{zetaCma})$.
Hence, the set of poles of $\zeta_A$ $($in $\Ce)$ and the residue of $\zeta_A$ at $s=D$ are respectively given by
\begin{equation}\label{2.1.6}
\po(\zeta_A)=(D+\mathbf p{\I}\Ze)\cup\{0\},\ \ \res(\zeta_A,D)=\frac{1-ma}{DT}\left(\frac{1-ma}{2(m-1)}\right)^{D-1},
\end{equation}
where $T:=\log(1/a)$
and $\mathbf p:=2\pi/T$ is the oscillatory period of $C^{(m,a)}$ $($in the sense of \cite{lapidusfrank12}$)$. Finally, each pole in $\po(\zeta_A)$ is simple.
\end{prop}

\subsection{Maximally hyperfractal strings}\label{hyperfratalsr}

We say that a bounded fractal string $\mathcal L=(\ell_j)_{j\in\eN}$ is {\em maximally hyperfractal}, if the corresponding geometric zeta function $\zeta_{\mathcal L}(s):=\sum_{j=1}^\ty\ell_j^s$ is such that
every point on the critical line $\{\re s=D\}$ is a nonremovable singularity of the corresponding geometric zeta function $\zeta_{\mathcal L}$.

\begin{theorem}\label{hyper} Let $D$ be a given real number in $(0,1)$.
Then there is an explicitly constructible bounded fractal string $\mathcal L=(\ell_j)_{j\in\eN}$  which is maximally hyperfractal and such that $\ov\dim_B\mathcal L=D$. The corresponding subset $A_{\mathcal L}:=\{\sum_{j\ge k}\ell_j:k\in\eN\}$ of the real line is also maximally hyperfractal and $\ov\dim_BA_{\mathcal L}=D$; see \S\ref{zeta_s}.
\end{theorem}

\begin{proof}
The fractal string $\mathcal L$ is obtained as a (suitably defined)  union of an infinite sequence of bounded fractal strings ${\mathcal L}_k:=(\ell_{kj})_{j\ge1}$, with carefully chosen values of the parameters $m_k$ and $a_k$ appearing in Definition \ref{Cma}, and where $(c_k)_{k\ge1}$ is an appropriate summable sequence of positive real numbers. For example, we may choose $c_k:=2^{-k}$, for every $k\ge1$. 

More precisely, we define ${\mathcal L}_k$, where $k\in\eN$, as the bounded fractal string corresponding to the generalized Cantor set $c_kC^{(m_k,a_k)}$, where $(m_k)_{k\in\eN}$ is any increasing sequence of integers, with $m_1\ge2$, and $a_k=m_k^{-1/D}$. (In other words, $c_kC^{(m_k,a_k)}$ is a geometric realization of $\mathcal L_k$, contained in the interval $[0,c_k]$.) Note that $C^{(m_k,a_k)}$ is well defined, since $m_ka_k=m_k^{1-1/D}<1$. Furthermore, $|{\mathcal L}_k|_1=c_k$ and, in light of Proposition \ref{Cmap}, we have that 
$
\dim_B{\mathcal L}_k=\dim_B C^{(a_k,m_k)}=\log_{1/a_k}m_k=D,
$
for any $k\in\eN$. In other words, the geometric zeta function of each fractal string $\mathcal L_k$ has the same critical line $\{\re s=D\}$.

Let 
${\mathcal L}:=\bigsqcup_{k=1}^\ty {\mathcal L}_k$ be the (disjoint) union 
of the sequence of bounded fractal strings ${\mathcal L}_k$. Since 
$
|\mathcal L|_1=\sum_{k=1}^\ty|\mathcal L_k|_1=\sum_{k=1}^\ty c_k<\ty,
$ 
we conclude that the fractal string $\mathcal L$ is bounded. 

Observe that the oscillatory period of ${\mathcal L}_k$, 
which is defined by
${\mathbf p}_k:=\frac{2\pi}{\log(1/a_k)}$,
provides valuable information about the density of the set $D+{\mathbf p}_k\Ze\I$ of principal complex dimensions of ${\mathcal L}_k$ on the critical line $\{\re s=D\}$; see Eq.\  \eqref{2.1.6} in Proposition \ref{Cmap}. More specifically,
since $a_k=m_k^{-1/D}\to0$ as $k\to\ty$, we see that for the set of principal complex dimensions of the generalized Cantor string ${\mathcal L}_k$ (i.e, the set of the principal poles of $\zeta_{{\mathcal L}_k}$),
$
\dim_{PC}{\mathcal L}_k=\dim_{PC} C^{(m_k,a_k)}=D+{\mathbf p}_k\Ze\I,
$
becomes denser and denser on the critical line, as $k\to\ty$, because then  the oscillatory period 
$
{\mathbf p}_k:=\frac{2\pi}{\log(1/a_k)}=\frac{2\pi D}{\log m_k}
$
of $\mathcal L_k$ tends to zero as $k\to\ty$, since $m_k\to\ty$.
Thus, the distance zeta function of the fractal string ${\mathcal L}:=\sqcup_{k=1}^\ty{\mathcal L}_k$ will have 
$
D+\Big(\bigcup_{k=1}^\ty{\mathbf p}_k\Ze\Big)\I
$
as a set of nonremovable singularities, which is densely packed on the critical line 
$\{\re s=D\}=D+\eR\I$,  
since the set 
$\cup_{k=1}^\ty{\mathbf p}_k\Ze$
is clearly dense in $\eR$.

Since we have a dense set of nonremovable singularities on the critical line, then in fact, each point on the line is a nonremovable singularity. 
To see this, let us reason by contradiction. Assume that there exists $s_0$ on the critical line, which is a removable singularity. By definition, and upon resolution of the singularity, this means that there exists an open disk $U:=B_\rho(s_0)$ in $\Ce$ centered at $s_0$ such that the fractal zeta function $\zeta_{\mathcal L}$ is holomorphic in the 
open disk $U$. Therefore, if $I:=U\cap \{\re s=D\}$ is the corresponding open interval along the critical line, then $\zeta_{\mathcal L}$ cannot have any nonremovable singularity in the open
interval $I$. However, this is clearly impossible, since it would then contradict the fact that the set of nonremovable singularities of $\zeta_{\mathcal L}$ is dense in the critical line. 




Note that, in essence, the above argument shows that the set of nonremovable singularities of $\zeta_{\mathcal L}$ is closed in $L:=\{\re s=D\}$, and since it is dense in $L$, it must then coincide with all of $L$.

In conclusion, we deduce that the entire critical line $\{\re s=D\}$ consists of nonremovable singularities of $\zeta_{\mathcal L}$, which means that $\mathcal L$ is maximally hyperfractal. 
In light of Eq.\  \eqref{cantor_string} and the text following it, the subset $A_{\mathcal L}$ of $\eR$ corresponding to the fractal string $\mathcal L$ is also maximally hyperfractal, as desired.
\end{proof}

\subsection{Maximal hyperfractals in higher-dimensional Euclidean spaces}\label{hyperfratalsn}
The aim of this subsection is to show that, given a maximal hyperfractal set $A$ in $\eR^N$, the sets of the form $A\times[0,1]^m$ will also be maximally hyperfractal for any positive integer $m$. The main result is stated in Theorem \ref{mh} below. It will enable us, in particular, to obtain an $N$-dimensional analog 
of Theorem \ref{hyper}; see Corollary \ref{hypN}.
\medskip

\begin{lemma}\label{line_sing}
Assume that $f=f(s)$ is a fractal zeta function such that $D_{\rm hol}(f)\in\eR$ and the corresponding critical line $\{\re s=D_{\rm hol}(f)\}$ of holomorphic continuation
$\{\re s=D_{\rm hol}(f)\}$ consists entirely of nonremovable singularities. Assume that $g=g(s)$ is holomorphic on $\{\re s>\a\}$,
where $\a<D_{\rm hol}(f)$. Then $D_{\rm hol}(f+g)=D_{\rm hol}(f)$ and the holomorphy critical line $\{\re s=D_{\rm hol}(f+g)\}$ of $f+g$ also consists entirely of nonremovable singularities.
\end{lemma}

\begin{proof} Since $f$ is holomorphic on $\{\re s>D_{\rm hol}(f)\}$, and by definition, the holomorphicity lower bound $D_{\rm hol}(f)$ is optimal (i.e., it is the infimum of all $\b\in\eR$ such that $f$ is holomorphic on $\{\re s>\b\}$), it then follows that $D_{\rm hol}(f)=D_{\rm hol}(f+g)$; indeed, $g$ is holomorphic on $\{\re s>\a\}\supseteq\{\re s>D_{\rm hol}(f)\}$.

In order to prove the second claim, we argue by contradiction and assume that some $s_0\in\Ce$ with $\re s_0=D_{\rm hol}(f+g)$ is a removable singularity of $f+g$. Then, since $g$ is holomorphic at $s_0$ (because $\a<D_{\rm hol}(g)$), it would follow that $s_0$ is a removable singularity of the function $f=(f+g)-g$ as well. However, this would contradict the assumption according to which the holomorphy critical line $\{\re s=D_{\rm hol}(f)\}$ consists of nonremovable singularities. 
\end{proof}

\begin{theorem}\label{mh}
Assume that $A$ is a maximally hyperfractal subset of $\eR^N$ and let $m\in\eN$.
Then the `fractal grill' $A\times[0,1]^m$ is also maximally hyperfractal.
\end{theorem}

\begin{proof}
By \cite[Thm.\ 3.15$(a)$]{dtzf}, we can write
$
\zeta_{A\times[0,1]^m}(s)=\zeta_A(s-m)+g(s)
$
for $\re s>\ov\dim_BA+m$,
where 
$
g(s):=\sum_{k=1}^m\binom mk\zeta_A(s-m+k)
$
is holomorphic on $\{\re s>\ov\dim_BA+m-1\}$.
By hypothesis, the holomorphy critical line of the function $f(s):=\zeta_A(s-m)$ is the vertical line
$\{\re s=\ov\dim_BA+m\}$ and consists entirely of nonremovable singularities. On the other hand, the function $g(s)$ is holomorphic on $\{\re s>\a:=\ov\dim_BA+m-1\}$, since this is the case of the functions $\zeta_A(s-m+k)$ for $k=1,2,\dots,m$.
(Here, we have also used the easily verified fact that $\ov\dim_BA$ does not depend on $N$, the embedding dimension.) Since $\a<\ov\dim_BA+m$, the claim now follows from Lemma \ref{line_sing}.
\end{proof}


\begin{cor}\label{hypN}
Let $N\in\eN$. Then, for any $D\in(N-1,N)$, there is an explicitly constructible maximally hyperfractal subset $A$ of $\eR^N$ such that $\dim_BA= D$.
\end{cor}



\bigskip
 

{ 

} 



\end{document}